\documentclass[aps,pra,reprint,groupedaddress]{revtex4-1}
\usepackage{dcolumn}
\usepackage{bm}
\usepackage[dvips]{graphicx}
\usepackage{amsthm}
\usepackage[cmex10]{amsmath}
\usepackage{amssymb}

\newcommand{\cpq}{C_{\psi\phi}}

\newcommand{\ket}[1]{\left| #1\right>} 
\newcommand{\hH}{\mathcal H}

\newcommand{\tr}{{\rm Tr}} 
\newcommand{\od}[2]{\frac{d #1}{d #2}} 
\newcommand{\real}{\mathbb R} 
\newcommand{\natn}{\mathbb N} 
\newcommand{\zah}{\mathbb Z} 
\newcommand{\comp}{\mathbb C} 

\newtheorem{Thm}{Theorem}
\newtheorem{Lem}{Lemma}

\newtheorem{Prop}{Proposition}

\begin{document}


\title{Asymptotic Compatibility between LOCC Conversion and Recovery}


\author{Kosuke Ito$^{1}$~~~Wataru Kumagai$^{1,2}$~~~Masahito  Hayashi$^{1,3}$\\
\textit{${}^1$Graduate School of Mathematics, Nagoya University, Japan, \\${}^2$
Faculty of Engineering, Kanagawa University, Japan, \\${}^3$
Centre for Quantum Technologies, National University of Singapore, Singapore}}



\begin{abstract}
Recently, entanglement concentration was explicitly shown to be irreversible.
However, it is still not clear what
kind of states can be reversibly converted in the asymptotic setting by LOCC
when neither the initial nor the target state is maximally entangled.
We derive the necessary and sufficient condition for the reversibility of 
LOCC conversions between two bipartite pure entangled states 
 in the asymptotic setting.
In addition, we show that conversion can be achieved perfectly with only local unitary operation under such condition except for special cases.
Interestingly, our result implies that an error-free reversible conversion is asymptotically possible even
between states whose copies can never be locally unitarily equivalent with any finite numbers of copies, although such a conversion is impossible
in the finite setting.
In fact, we show such an example.
Moreover, we establish how to overcome the irreversibility of LOCC conversion in two ways.
As for the first method, we evaluate
how many copies of the initial state is to be lost 
to overcome the irreversibility of LOCC conversion.
The second method is to add a supplementary state appropriately, which also works for LU conversion unlike the first method.
Especially, for the qubit system, any non-maximally pure entangled state can be a universal resource for the asymptotic reversibility when copies of the state is sufficiently many.
More interestingly, our analysis implies that far-from-maximally entangled states can be better than nearly maximally entangled states as this type of resource. This fact brings new insight to the resource theory of state conversion.

\end{abstract}


\maketitle



\section{Introduction}
Entangled states are used as resources for many quantum information processes.
When the initial entangled state is different from the desired form and we are not allowed to apply the global operation,
we need to convert the given initial state by local operations and classical communications (LOCC).
This type of conversion is called LOCC conversion.
Many conventional researches deal with LOCC conversions whose target states are maximally entangled states.
However, the most preferable entangled state depends on the type of the information processes to be applied.
For example, measurement based quantum computation \cite{gross09:_most_quant_states_are_too} and quantum channel estimation \cite{hayashi11:_compar_cramer_rao}
require entangled states that are not necessarily maximally entangled
while maximally entangled states are used as typical resource of entanglement.
Moreover, Ishizaka and Hiroshima \cite{ishizaka09:_quant} revealed that
the optimal shared entangled state is not necessarily a maximally entangled state in their port-based teleportation.
In such a situation, it is required to consider case where both the initial and the target states are not necessarily maximally entangled.

Bennett et. al. \cite{bennett96:_concen} studied the asymptotic conversion between the multiple-copy states of two distinct pure entangled states, which are not necessarily maximally entangled.
The optimal conversion rates are given by the ratio between von Neumann entropies $S_\psi$ and $S_\phi$ of the reduced density matrices of the initial state $\psi$ and the target state $\phi$.
Since the opposite conversion rate is the inverse of the original conversion rate, this kind of conversion was seemed to be reversible, as pointed out in \cite{jonathan99:_minim_condit_local_pure_state_entan_manip}\cite{vidal01:_irrev_asymp_manip_entan}\cite{yang05:_irrev_all_bound_entan_states}\cite{acin03}\cite{bennett00:_exact}.
However, two of the authors \cite{kumagai13:_entan_concen_irrev} explicitly revealed that this kind of conversion is irreversible
in the case of entanglement concentration, i.e., 
the case when the target entangled state is maximally entangled,
although Hayden and Winter \cite{hayden03:_commun} and Harrow and Lo \cite{harrow04:_tight_lower_bound_class_commun} implicitly suggested this fact.
This problem has not been discussed when 
the initial and target states are not maximally entangled.
Recently, two
of the authors \cite{w.13:_asymp_class_locc_conver_its} investigated the second order asymptotics
and derived the second-order optimal LOCC conversion
rate between general pure states, which clarifies
the relation between the accuracy and the asymptotically
optimal conversion rate up to the second order.
However, the paper \cite{w.13:_asymp_class_locc_conver_its} did not consider the reversibility.
That is, it is still unsolved what
kind of states can be reversibly converted in the asymptotic setting by LOCC
when neither the initial nor the target state is maximally entangled.
To clarify the reversibility, this paper studies the compatibility between the LOCC conversion and the recovery operation when the initial and the target states are given as respective number of copies of
an arbitrary pure entangled state $\psi$ and another arbitrary 
pure entangled state $\phi$ on bipartite system $\hH_A\otimes\hH_B$.

Next, we consider to restrict our operations to local unitary operations (LU conversion), which are contained in LOCC conversions. This analysis is not only an examination whether or not the equivalence between LU convertibility and LOCC reversibility is true even for the asymptotic setting, but also useful when the reversibility is required primally even for the finite setting. We derive the formula of the error of LU conversion except for the lattice cases, with a new method of asymptotic analysis of probability distributions. We also give numerical calculations of the error of LU conversion for some examples of pure entangled states to compare with the asymptotic error. Especially, we provide an example of LU convertible pair of states by numerical calculation, and examine how its error converges to nearly $0$.

Further, when the the asymptotic reversibility of LOCC conversions is impossible,
we also consider how to realize the asymptotic reversibility by modifying this conversion 
in two different ways.
As the first method, we allow to decrease the required number of copies in the recovery process.
In this setting, given $n$ copies of the pure entangled state $\psi$ initially,
we find that when the decreasing number is in the order $n^\gamma$ with $\gamma >1/2$,
both of conversion and recovery can be simultaneously realized in the asymptotic setting.
As the second method,
we consider to add a supplemental resource to satisfy the 
necessary and sufficient condition for the asymptotic reversibility.
Then, we find that 
a state which is far from the maximally entangled state is more useful as a supplemental resource.

The paper is organized as follows.
At first, we show the necessary and sufficient condition for the asymptotic reversibility of LOCC conversions between two bipartite pure entangled states
when the number of copies to be recovered is restricted to be the same as the initial number in Sec. \ref{noloss-sec0}.
Next, we restrict our operations to
local unitary operations (LU conversion) in Sec. \ref{LU-main},
which are contained in LOCC conversions.
This analysis is useful when the reversibility is required primally.
We give numerical calculations of the error of LU conversion for some examples of pure entangled states.
Especially, we provide an example of LU convertible pair of states by numerical calculation in Sec. \ref{LU-example}.
Further, we consider how to overcome the asymptotic irreversibility of LOCC conversions in two ways in Sec. \ref{overcome_mother}.
Firstly, we relax the constraint so that the number of copies to be recovered can be smaller than the initial number in Sec. \ref{someloss-sec}.
We show the tight lower bound for the order of the difference between the initial and recovered number needed to overcome the irreversibility of LOCC conversion.
Secondly, we consider how to realize the asymptotic reversibility by adding a supplemental resource
when the asymptotic reversibility of LOCC conversions is impossible in Sec. \ref{supplement_method}.
That is, we clarify how many sacrificed copies are required for asymptotically error-free recovery.
Finally, the summary of our work and future perspectives are provided in Sec. \ref{conclusion-sec}.
All of theorems given in the main body will be essentially shown in Appendices.
\section{Compatibility between Conversion and Recovery without Any Loss}\label{noloss-sec0}
\begin{figure}[!t]
\centering
\includegraphics[clip ,width=2.8in]{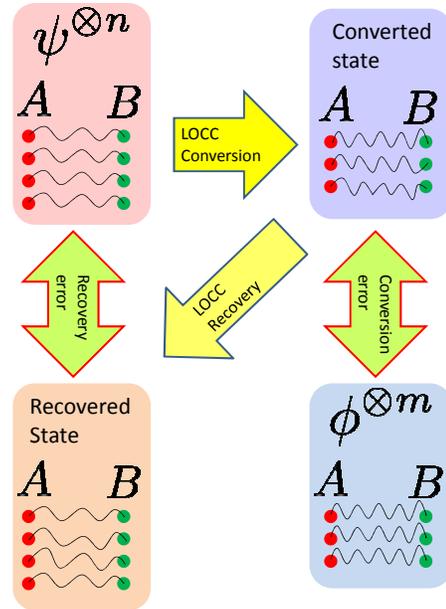}
\caption{The schematic diagram of LOCC conversion and its recovery operation. The initial state which consists of $n$ copies of a pure entangled state $\psi$ is converted to the target state which consists of $m$ copies of another pure entangled state $\phi$.}
\label{figure1}
\end{figure}
We consider the LOCC conversion and the recovery to the same number of copies of the initial state as the initial number,
as described in Fig. \ref{figure1}.
To investigate the compatibility between LOCC conversion and recovery,
we consider the minimum conversion-recovery error (MCRE) defined as
\begin{align}
 &\delta_n(\psi,\phi) :=\nonumber\\
&\min_{m\in\natn , C,D:\text{LOCC}} B(C(\psi^{\otimes n}),\phi^{\otimes m})+B(\psi^{\otimes n},D\circ C(\psi^{\otimes n})),
\end{align}
where $B$ is the Bures distance defined as $B(\psi, \phi) = \sqrt{1 - F(\psi,\phi)}$, $F$ denotes the fidelity, $C$ and $D$ are LOCC conversion and recovery operations, respectively.
The limit $\lim_{n\to \infty}\delta_n(\psi,\phi)$ represents the asymptotic compatibility between the two operations because its convergence to zero means that both operations can be perfectly accurately done in the asymptotic setting.
On the other hand, when it does not go to zero, we have to consider a trade-off between
the convertibility and the reversibility 
even in the asymptotic setting.

To investigate the asymptotic behavior of $\delta_n(\psi,\phi)$, we focus on the asymptotic expansion of the number $m_n=S_\psi/S_\phi n + b n^{\gamma} + o(n^{\gamma})$ of copies of the target state $\phi$ when the initial state is $\psi^{\otimes n}$.
First of all, when $b<0$ and $\gamma>1/2$, the conversion is possible since the target number of copies is small, however the recovery is no longer impossible because of fatal lost of number of copies.
Oppositely, when $b>0$ and $\gamma>1/2$, the conversion itself is impossible since the target number of copies is too large.
The remaining case where $\gamma = 1/2$ requires more delicate analysis, so-called the second order analysis.
In the second order analysis of LOCC conversion, it was shown that 
the following quantity plays an important role 
\cite{w.13:_asymp_class_locc_conver_its}:
\begin{equation}
 C_{\psi, \phi} := 
\frac{S_{\psi}}{V_\psi}\left(\frac{S_{\phi}}{V_\phi}\right)^{-1},
\end{equation}
where 
$V_\psi := \tr\{(\tr_B\psi)(-\log(\tr_B\psi)-S_{\psi} )^2\}$.
In the following,
we call $C_{\psi\phi}$ the conversion characteristics.
Our result is that the conversion characteristics completely characterizes the asymptotic compatibility between LOCC conversion and recovery for pure bipartite entangled states as follows.
\begin{Thm}\label{noloss}
 Let $\hH_A$ and $\hH_B$ be finite-dimensional systems and $\psi$ and $\phi$ be arbitrary pure states on $\hH_A\otimes\hH_B$. Then
\begin{equation}
 \lim_{n\rightarrow\infty}\delta_n(\psi,\phi)=0,
\end{equation}
if and only if $\cpq = 1$.
\end{Thm}
The detail of the proof of Theorem \ref{noloss} is given in Appendix \ref{sa2}.
Now, one may feel that our result is strange in comparison to the
result by Bennett et. al. \cite{bennett96:_concen}, which states that
\begin{align}
\lim_{n\rightarrow\infty}\sup\left\{\frac{m_n}{n}|\min_{C:\text{LOCC}} B(C(\psi^{\otimes n}),\phi^{\otimes m_n})<\epsilon\right\}=\frac{S_\psi}{S_\phi}\label{benn}
\end{align}
holds for arbitrary $\epsilon >0$.
Because of (\ref{benn}), LOCC conversion is conventionally considered to be asymptotically reversible in the sense that the conversion $\psi^{\otimes n}\rightarrow \phi^{\otimes (S_\psi/S_\phi n+o(n))}\rightarrow \psi^{\otimes (n(=(S_\phi/S_\psi)(S_\psi/S_\phi)n)+o(n))}$ is possible asymptotically.
However, $o(n)$ is quite ambiguous in this context.
In fact, the conversion $\psi^{\otimes n}\rightarrow \phi^{\otimes S_\psi/S_\phi n+b\sqrt{n}+o(\sqrt{n})}$ is impossible if $\cpq\neq 1$ as we consider in the proof of Theorem \ref{noloss}.
That is because (\ref{benn}) does NOT mean that any number of the form $m_n= S_\psi/S_\phi n +o(n)$ can always be achievable, although it implies that the maximally achievable number $m_n$ of the copies have the form $m_n= S_\psi/S_\phi n +o(n)$.
In Sec. \ref{someloss-sec},
we deal with the case when we permit some loss in the number of copies to be recovered.
In this way, we correct the conventional naive sense of $o(n)$, simultaneously giving the first method to overcome the irreversibility, i.e., to allow to decrease the required number of copies in the recovery process.

e reversibility
As stated in Theorem \ref{noloss}, if the given initial state $\psi$ and the desired state $\phi$ have $\cpq\neq 1$,
there is no reversible LOCC conversion between them even for asymptotic sense.
Fortunately, however, such a conversion becomes possible by adding an appropriate state $\psi'$ to adjust $C_{\psi\otimes\psi',\phi}$ to $1$ as we examine in detail in Sec. \ref{supplement_method}, which is the second method to overcome the irreversibility.
In this sense, Theorem \ref{noloss} itself is useful to realize reversible conversion for desired pair of states.

\section{Conversion with Only Local Unitary Operation}
\subsection{The Asymptotic Formula of the Error under LU conversion}\label{LU-main}
Even if LOCC conversion and recovery are asymptotically compatible, it is not necessarily the case for finite number of copies in general.
If our operation is restricted to local unitary operations, its reversibility is perfectly guaranteed even for the non-asymptotic setting.
In some applications, the reversibility is more important
than the quality of conversion.
For example, consider the following protocol.
At first, we encode message to shared entangled state.
Then, we ``hide'' the existence of entanglement behind the thermal noise 
by converting the entanglement half to a state similar to the environment 
so that eavesdropper can not notice the existence.
This conversion must be reversible for decoding.
To guarantee the perfect recoverability of the message,
we need to restrict our operations to LU conversions.
We define the error $\epsilon_n(\psi, \phi)$ of LU conversion from $\psi$ to $\phi$ as
\begin{align}
& \epsilon_n(\psi, \phi):=\nonumber\\
 & \min
\left\{
B((U_A\otimes U_B)\psi^{\otimes n}, \phi^{\otimes m})
\left|
\begin{array}{l}
\displaystyle
 m\in \natn\\
 U_A : \hH_A^{\otimes n}\rightarrow\hH_A'^{\otimes m}\\
 U_B : \hH_B^{\otimes n}\rightarrow\hH_B'^{\otimes m}\\
 U_A, U_B :\text{unitary}
\end{array}
\right.
\right\}.
\end{align}
This definition of the error measures the optimum performance under the optimization of the local unitary operation $U_A\otimes U_B$ and the number $m$ of copies of the target state $\phi$.

Similarly to the previous section, we note that $\epsilon_n(\psi,\phi)$ can be represented by Schmidt coefficients.
We can maximize the fidelity
between the target state and the converted state
in terms of local unitary operations $U_A$ and $U_B$ as follows (see Lemma \ref{vidalem} in Appendix \ref{defnotes}):
\begin{equation}
\max_{U_A, U_B : \text{unitary}}F((U_A\otimes U_B)\psi, \phi) = F(P_\psi^\downarrow, P_\phi^\downarrow).\label{unitary}
\end{equation}
Thus, the following holds:
\begin{equation}
\epsilon_n(\psi,\phi)^2 = 1 - \max_{m\in\natn}F(P_\psi^{n\downarrow}, P_\phi^{m\downarrow}), \label{classical}
\end{equation}
where $P_\psi$ and $P_\phi$ are the probability distributions composed of the squared Schmidt coefficients of $\psi$ and $\phi$ respectively, and $\downarrow$ denotes reordered version of the probability distribution in a descending manner (see also Appendix \ref{defnotes}).
Then, the problem is again reduced to the analysis on probability distributions.

Before we go further,
we mention the following excepted type of distribution for our analysis.
We call a distribution $P$ on ${\cal A}$
a lattice distribution when there exist $x,d\in \real$ such that 
$(\log P(a) - x)/d \in \zah$ for any $a \in {\cal A}$ with $P(a)>0$.
Then, the following lemma is the heart of our analysis for LU conversion.
\begin{Lem}\label{mainlem}
When $\left|m_n-\frac{S_{\psi}}{S_{\phi}}n\right|/\sqrt{n}\rightarrow\infty$, we have
\begin{align}
 \lim_{n\rightarrow\infty}F\left(P_{\psi}^{n\downarrow}, P_{\phi}^{m_n\downarrow}\right)=0.
\end{align}
Moreover, we assume that neither $P_\psi$ nor $P_\phi$ is lattice distributions.
Then the following holds when $\frac{m_n-\frac{S_\psi}{S_\phi}n}{\sqrt{n}}\rightarrow b$.
\begin{align}
& \lim_{n\rightarrow\infty}F\left(P_{\psi}^{n\downarrow}, P_{\phi}^{m_n\downarrow}\right) \nonumber\\
=& 
\sqrt{\frac{2}{\cpq^{\frac{1}{2}}+\cpq^{-\frac{1}{2}}}}\exp \left[-\frac{b^2S_\phi^2}{4(1+\cpq)V_\psi}\right].\label{limfid}
\end{align}
\end{Lem}
This result is interesting itself at a point that it characterizes the distance of quite generic probability distributions 
only by using entropy and variance.
To make such analysis possible, we derive a new method of comparing probabilities by using strong large deviation by Bahadur and Rao \cite{bahadur60:_deviat_sampl_mean}.
At first, we can interpret the fidelity as the following expectation value.
\begin{align}
 F\left(P_{\psi}^{n\downarrow}, P_{\phi}^{m_n\downarrow}\right)&=\sum_{j} P^{n\downarrow}_{\psi} (j)
 \sqrt{\frac{P^{m_n\downarrow}_{\phi} (j)}{P^{n\downarrow}_{\psi} (j)}}\nonumber\\
 &=:E_{P^{n\downarrow}_{\psi}}\left[\sqrt{\frac{P^{m_n\downarrow}_{\phi}}{P^{n\downarrow}_{\psi}}}\right].\label{expect}
\end{align}
We have to deal with the ratio of probability $P_{\psi}^{n\downarrow}(i)$ to $P_{\phi}^{m_n\downarrow}(i)$ with the same number $i$.
To do so, we apply the strong large deviation to the uniform distribution, and give the  asymptotic expansion of the number of elements of the event $\{j| \log P^{n\downarrow}_{\psi(\phi)} (j) \ge - n S_{\psi(\phi)}+ \sqrt{n}t \}$ for each $t$, which roughly gives the order $i$ of the value $P^{n,\downarrow}_{\psi(\phi)}(i)$ corresponding to $t$ such that
\begin{align}
t_{\psi(\phi), n}(i):=t=\frac{\log P^{n\downarrow}_{\psi(\phi)} (i) + n S_{\psi(\phi)}}{\sqrt{n}}.
\end{align}
Then, we evaluate
$
 t_{\phi, m_n}(t_{\psi, n}^{-1}(t))
$
for each $t$ by solving the following equation about $x$ by using the above mentioned expansion:
\begin{align}
 P_C\{j| t_{\phi,m_n}(j) \ge x \}=P_C\{k| t_{\psi,n}(k) \ge t \},
\end{align}
where $P_C$ is the counting measure.
In this way, we make $P_\phi^{m_n\downarrow}(i)$ a function of only $t_{\psi,n}(i)$ up to irrelevant terms.
Finally, thanks to this function, we can calculate the expectation (\ref{expect}) by applying the central limit theorem to $t_{\psi,n}$.
The detailed proof of Lemma \ref{mainlem} is in Appendix \ref{sa32}.

Now, we reinterpret this lemma as for LU conversion.
We call a state $u\in\hH_A\otimes\hH_B$ a lattice state if $P_u$ is a lattice distribution.
In this case, because of \eqref{limfid} of Lemma \ref{mainlem}, the fidelity
$F\left(P_{\psi}^{n\downarrow}, P_{\phi}^{m_n\downarrow}\right)$
has asymptotic maximum  value $\sqrt{\frac{2}{{\cpq}^{\frac{1}{2}}+{\cpq}^{-\frac{1}{2}}}}$ when $b=0$.
Hence, 
the optimal number $m_n$ of initial copies of $\phi$ 
has the asymptotic expansion as ${S_\psi}/{S_\phi}n+o(\sqrt{n})$.
Therefore, combining (\ref{classical}), we have the following Theorem \ref{main}.

\begin{Thm}\label{main}
The following holds when neither $\psi$ nor $\phi$ is lattice states.
\begin{align}
\lim_{n\rightarrow\infty}\epsilon_n(\psi, \phi)^2= 1-\sqrt{\frac{2}{{\cpq}^{\frac{1}{2}}+{\cpq}^{-\frac{1}{2}}}}.
\label{limep}
\end{align}
Moreover, the optimal number of initial copies of $\phi$ is $\frac{S_{\psi}}{S_{\phi}}n + o(\sqrt{n})$.
\end{Thm}
If the conversion characteristics $\cpq$ is close to $1$, $\lim_{n\rightarrow\infty}\epsilon_n(\psi, \phi)$ is close to $0$, i.e., the state $\phi^{\otimes m_n}$ can be asymptotically precisely approximated by LU conversion.
This fact guarantees that even if $\cpq$ is slightly deviated from $1$, $\lim \delta_n(\psi,\phi)$ is not far from $0$.
 
 However, the above mentioned rough estimation does not work well for lattice distributions, which makes the analysis more difficult.
For every lattice state $\psi$, there are non-lattice states $\psi_{\alpha}$ as close to it as we want, where $\psi_{\alpha}$ is parameterized by $\alpha$ such that $\psi_{\alpha}\rightarrow \psi(\alpha\rightarrow 0)$.
 Even for lattice states, we would have the same formula as \eqref{limep}
if we could prove
 \begin{align}
  \lim_{\alpha\rightarrow 0}\lim_{n\rightarrow\infty}\epsilon_n(\psi,\psi_{\alpha})=0.\label{8-6-1}
 \end{align}
However, we have a counter example as explained latter.

Instead of Theorem \ref{main}, 
for lattice states,  we have the following lower bound, 
which is proven by some estimation and the central limit theorem (see Lemma \ref{upper-thm} in Appendix \ref{sa31}).
\begin{Thm}\label{low-err}
 The following holds.
\begin{align}
 \lim_{n\rightarrow\infty}\epsilon_n(\psi, \phi)^2\geq 1-\sqrt{\frac{2}{{\cpq}^{\frac{1}{2}}+{\cpq}^{-\frac{1}{2}}}}.
\end{align}
\end{Thm}
Now, we give a counter example for \eqref{8-6-1}.
When $\psi$ is a maximally entangled state (which is a lattice state) and $\phi$ is not, $\cpq = \infty$ and hence $\lim_{n\rightarrow\infty}\epsilon_n(\psi,\phi)=1$ always holds by Theorem \ref{low-err}, which means that this LU conversion always has maximum error.
 Thus, for this $\psi$, we have $\lim_{\alpha\rightarrow 0}\lim_{n\rightarrow\infty}\epsilon_n(\psi,\psi_{\alpha})=1\neq0$, which gives a counter example for \eqref{8-6-1}.
 Theorem \ref{low-err} also implies that lattice states do not have advantage in LU convertibility in comparison to non-lattice states.

In the finite setting, the fidelity becomes $1$ only 
in a limited case, i.e., the case when 
$P_{\psi}^{n\downarrow}=P_{\phi}^{m\downarrow}$
with a certain $m$.
However, due to this theorem,
as the number $n$ goes to infinity,
the fidelity approaches $1$,
i.e., these two states are inter-convertible by LU conversion
under the weaker condition $\cpq = 1$.
In fact, such a non-trivial example is numerically given in the next subsection.

\subsection{Examples in the System where Alice and Bob Have each Two Qubit}\label{LU-example}
\begin{figure}[!t]
\centering
\includegraphics[clip, width=2.8in]{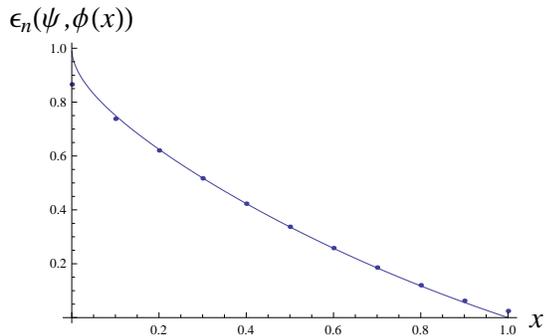}
\caption{The solid line is the graph of $[1-\text{sech}^{\frac{1}{2}}\frac{1}{2}\log C_{\psi,\phi(x)}]^{1/2}$, which equals to the asymptotic error $\lim\epsilon_n(\psi,\phi(x))$ of LU conversion according to Theorem \ref{main}. We see that $\lim\epsilon_n(\psi,\phi(x))$ is $1$ with $x=0$, and is almost equal to $0$ with $x=1$ because $C_{\psi,\phi(0)}=0$ and $C_{\psi,\phi(1)}\approx 1$. 
The asymptotic error $\lim\epsilon_n(\psi,\phi(x))$ takes various values in proportion to $x$. The dots are the result of numerical calculation of $\epsilon_{3000}(\psi,\phi(x))$ for $x=0.1j\;(j=0,1,\dots,10)$. Indeed, they are close to the asymptotic curve given as the solid line.}
\label{graph0}
\end{figure}
\begin{figure}[!t]
\centering
\includegraphics[clip, width=2.8in]{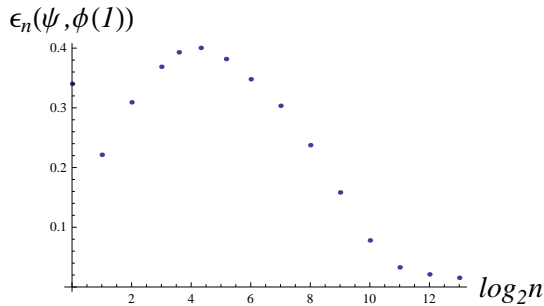}
\caption{The plot of numerical calculation of $\epsilon_n(\psi,\phi(1))$ vs. $n$ from $n=1$ to $n=2^{13}$. We see that the error $\epsilon_n(\psi,\phi(1))$ converges to almost $0$ because $C_{\psi,\phi(1)}$ is close to $1$, just as Theorem \ref{main} states.}
\label{graph1}
\end{figure}
Now, we give examples
in a bipartite system $\hH_A\otimes\hH_B$, where $\hH_A = \hH_B = \comp^2\otimes\comp^2$, i.e., both Alice and Bob have each two qubits and share entangled states between them.
 Let $\{\ket{0},\ket{1}\}$ be a computational basis of each qubit.
At first, in order to see the asymptotic behavior of the error of the LU conversion given in Theorem \ref{main}, we define the initial and target pure states $\ket{\psi}\in\hH_A\otimes\hH_B$ and $\ket{\phi(x)}\in\hH_A\otimes\hH_B$ for $0\leq x\leq 1$ as
\begin{align}
&\ket{\psi}:=\sqrt{0.0048}\ket{00}_A\otimes\ket{00}_B 
+\sqrt{0.4752}\ket{01}_A\otimes\ket{01}_B \nonumber\\
&\hspace{15pt}+\sqrt{0.0052}\ket{10}_A\otimes\ket{10}_B
+\sqrt{0.5148}\ket{11}_A\otimes\ket{11}_B,\\
&\ket{\phi(x)}:= \sqrt{(-ax + 0.5)(-bx + 0.5)}\ket{00}_A\otimes\ket{00}_B \nonumber\\
&\hspace{30pt}+\sqrt{(-ax + 0.5)(bx + 0.5)}\ket{01}_A\otimes\ket{01}_B \nonumber\\
&\hspace{30pt}+\sqrt{(ax+0.5)(-bx + 0.5)}\ket{10}_A\otimes\ket{10}_B\nonumber\\
&\hspace{30pt}+\sqrt{(ax+0.5)(bx + 0.5)}\ket{11}_A\otimes\ket{11}_B,
\end{align}
where $a:=0.225$ and $b:= 0.1996180626854719$.
In FIG. \ref{graph0}, the solid line is the asymptotic error $\lim\epsilon_n(\psi,\phi(x))$ given in Theorem \ref{main} as a function of $x$,
 where we see that $\lim\epsilon_n(\psi,\phi(x))=1$ with $x=0$, and $\lim\epsilon_n(\psi,\phi(x))\approx 0$ with $x=1$ because
$\phi(0)$ is a maximally entangled state, and $C_{\psi,\phi(1)}=1+1.11\times 10^{-15} \approx 1$.
We can also see that the limit of the error can take various values dependently of the target states.
As for the error for the non-asymptotic setting, the dots in FIG. \ref{graph0} are the numerical
results of $\epsilon_{3000}(\psi,\phi(x))$ for $x=0.1j\;(j=0,1,\dots,10)$. Indeed, we can see that the error for large $n = 3000$ is close to (\ref{limep}).
In particular, $\psi$ and $\phi(1)$ are obviously not locally unitarily equivalent, and
satisfy $C_{\psi,\phi(1)} \approx 1$ as mentioned above.
Hence, the pair $(\psi, \phi(1))$ is a non-trivial example of an asymptotically precisely LU convertible pair.
Thus, we next see the convergence of $\epsilon_n(\psi,\phi(1))$ to $0$ as $n$ goes to infinity in detail.
FIG. \ref{graph1} plots the numerical results of $\epsilon_n(\psi,\phi(1))$ for several cases of $n$.
We can see that indeed the error is converging to $0$ as $n$ increases.

\section{How to Overcome the Asymptotic Irreversibility} \label{overcome_mother}
 In this section, we establish two methods to overcome the irreversibility of LOCC conversion.
 The first is to decrease the required number of copies in the recovery process, and the second is to add an appropriate supplementary state.
 Note that the first method does not work for LU conversion, since LU convertibility is not affected at all if the required number is decreased in the recovery process, while the second works similarly as for LOCC, since it is based on the same reversible condition $\cpq=1$.
\subsection{Recovery with Some Loss}\label{someloss-sec}
\begin{figure}[!t]
\centering
\includegraphics[clip, width=2.8in]{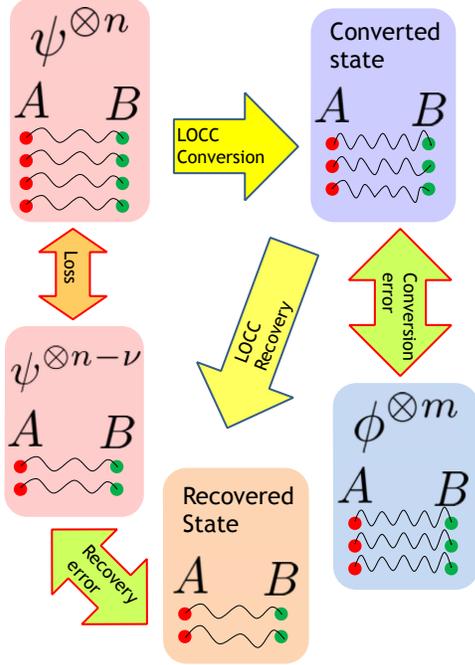}
\caption{The schematic diagram of LOCC conversion and its recovery with some loss in the number of copies of the recovered state. 
The recovered number $n-\nu$ of copies of $\psi$ 
is slightly smaller than
the initial number $n$ of copies of $\psi$ 
due to the loss.}
\label{figure2}
\end{figure}
In Sec. \ref{noloss-sec0},
the number of the copies of the state to be recovered has been restricted to the same number as the initial number $n$.
In this case, 
we have already showed that perfect conversion and recovery 
are not simultaneously realized even in the asymptotic setting if $\cpq \neq 1$.
Then, we relax this condition so
that the number $n-\nu$ of the copies to be recovered can be smaller than the initial number $n$
as is described in Fig. \ref{figure2}.
That is, we permit some loss $\nu\geq 0$ in the number of copies to be recovered.
Then, under this condition, we investigate the compatibility between conversion and recovery
when $\cpq \neq 1$ moreover.
In order to investigate the order of the amount of loss to realize 
perfect conversion and recovery simultaneously,
we consider the following quantity.
\begin{align*}
R(\psi,\phi)
:=
\inf_{\gamma}
\left\{
\gamma
\left|
\begin{array}{l}
\displaystyle \min_{m\in\natn , C,D:\text{LOCC}} [
B(C(\psi^{\otimes n}),\phi^{\otimes m}) \\
\quad +B(\psi^{\otimes (n-n^{\gamma})},D\circ C(\psi^{\otimes n})) ]\\
\to 0 \qquad (n\to \infty)
\end{array}
\right.
\right\}.
\end{align*}
This quantity is the infimum of the order of the sufficient amount of loss for asymptotic compatibility between conversion and recovery.

Now, we examine how large loss is necessary and sufficient to overcome irreversibility.
At first, we have to success the forward conversion.
When $\cpq\neq 1$, we can only make $S_\psi/S_\phi n - n^{\gamma}$($\gamma >1/2$) copies of $\phi$ from $n$ copies of $\psi$ via LOCC.
Thus, it is enough to consider the possibility of the recovery after the conversion to $S_\psi/S_\phi n - n^{\gamma}$ copies of $\phi$ with $\gamma>1/2$.
Then, we cannot recover $n-\beta \sqrt{n}$ copies after such conversion, applying the same analysis to the recovery process as that of the forward conversion.
However, we can recover $n-n^{\gamma}$ copies for any $\gamma>1/2$ because of the possibility of the conversion to $S_\psi/S_\phi n - n^{\gamma}$ copies with the same $\gamma$.
In such a way of estimation, we conclude that even when $\cpq\neq1$, asymptotic compatibility is overcome if and only if the recovery number is $n-n^{\gamma}$ with $\gamma>1/2$.
In other words, we obtain the following theorem. The detail of the proof is in Appendix \ref{Proof_HT3}.
\begin{Thm}\label{HT3}
The relation $R(\psi,\phi)=\frac{1}{2}$ holds when $\cpq \ne 1$.
More precisely, the following holds for an arbitrary number $\beta$ when $\cpq \neq 1$;
\begin{align}
 &\lim_{n\rightarrow\infty}\min_{m\in\natn , C,D:\text{LOCC}} 
 B(C(\psi^{\otimes n}),\phi^{\otimes m})\nonumber\\
&\quad+B(\psi^{\otimes (n-\beta n^{\frac{1}{2}})},D\circ C(\psi^{\otimes n})) \nonumber\\
\neq& 0.
\end{align}
\end{Thm}

When $\cpq\neq 1$, Theorem \ref{HT3} implies that LOCC conversion
is reversible if and only if the order of the amount of loss is strictly greater than the square root of the initial number of copies. 
Recall that $o(n)$ in the conversion $\psi^{\otimes (n+o(n))}\rightarrow \phi^{\otimes (S_\psi/S_\phi n+o(n))}\rightarrow \psi^{\otimes (n+o(n))}$ is naively considered in the context of the conventional asymptotic reversibility as was mentioned at the end of Sec. \ref{noloss-sec0}.
Now, we correct it as the conversion $\psi^{\otimes n}\rightarrow \phi^{\otimes (S_\psi/S_\phi n-O(n^{\gamma}))}\rightarrow \psi^{\otimes (n-n^{\gamma})}$ is asymptotically possible if and only if $\gamma>1/2$.
Such difference is effective as $n$ is not so large, which is realistic.
However, note that the significant condition $\cpq =1$ never occurred in the conventional way.
We can make use of this criteria to safely preserve the entanglement by making $C_{\psi\otimes \psi', \phi}=1$ as was mentioned in Sec. \ref{noloss-sec0}.
\subsection{Adding a Supplemental Resource}\label{supplement_method}
As for the second method to overcome the irreversibility, we consider the strategy to use supplementary states.
This method works for LU conversion as well as LOCC conversion unlike the first method.
That is, instead of the original conversion $\psi\to \phi$, we consider the conversion 
$\psi\otimes\psi'\to \phi$ by adjusting $C_{\psi\otimes\psi', \phi}$ to $1$.
To satisfy the condition $C_{\psi\otimes\psi', \phi}=1$,
$\psi'$ needs to satisfy the condition
$S_{\psi'}=\frac{S_\phi}{V_\phi} V_{\psi'}+V_{\psi}(\frac{S_\phi}{V_\phi}-\frac{S_\psi}{V_\psi})$.
The entropy $S_{\psi'}$ needs to be larger than $ \frac{S_\phi}{V_\phi} V_{\psi'}$
when $S_\psi/V_\psi < S_\phi/V_\phi$,
nevertheless it needs to be smaller than $ \frac{S_\phi}{V_\phi} V_{\psi'}$ when $S_\psi/V_\psi > S_\phi/V_\phi$.
That is, the ratio $\frac{S_{\psi'}}{V_{\psi'}}$ to be smaller than the ratio $\frac{S_\phi}{V_\phi}$
in the latter case.
In this case, the state $\psi'$ is further from the maximally entangled state.
So, such a state is usually thought to be useless because it has a weaker entanglement.
However, such a state is useful as a resource for the conversion 
$\psi\otimes\psi'\to \phi$ when $S_\psi/V_\psi > S_\phi/V_\phi$.
This fact brings new insight for the resource conversion theory.
In this application, a state $\psi'$ with a small ratio $\frac{S_{\psi'}}{V_{\psi'}}$ can be used for a supplemental resource dependently of the two states
$\psi$ and $\phi$.

Moreover, by modifying this application, any pure entangled state can be universally used for a supplemental resource for state conversion in the qubit system as follows.

\begin{Thm}\label{assist}
Let $\hH_A=\hH_B=\comp^2$. 
For an arbitrary non-maximally pure entangled state $\omega$ of $\hH_A\otimes\hH_B$,
there exists sufficiently large number $K$ such that if $k\geq K$, $\omega^{\otimes k}$ works as a universal resource for the asymptotic reversibility in the following sense.
For arbitrary two pure entangled states $\psi$ and $\phi$ of $\hH_A\otimes\hH_B$, 
there exists a pure state $\omega'$ of $\hH_A\otimes\hH_B$ 
 such that
 \begin{align}
  C_{\psi\otimes\omega^{\otimes k}, \phi\otimes\omega'^{\otimes k}}=1\label{supplement0}
 \end{align}
 holds.
\end{Thm}
This theorem is shown in Appendix \ref{Proof_assist}.
In this modified scheme, we prepare $k$ copies of the pure entangled state $\omega$
and obtain a byproduct state $\omega'^{\otimes k}$ in this conversion
because we consider the conversion $\psi\otimes\omega^{\otimes k}\to \phi\otimes\omega'^{\otimes k}$
instead of the original conversion $\psi\to\phi$.
Theorem \ref{assist} states that 
any non-maximally pure entangled state $\omega$ of $\hH_A\otimes\hH_B$
can be used as a supplemental resource for state conversion 
for any two states $\psi$ and $\phi$.
Here, we need to choose the number of copies $K$ sufficiently large,
however, the number $K$ does not depend on either the state $\psi$ or the state $\phi$.
Hence, $k$ copies of $\omega$ universally works for this kind of state conversion, i.e.,
the state $\omega^{\otimes k}$ can be regarded as a universal resource.
\begin{figure}[!t]
\centering
 \includegraphics[clip ,width=2.8in]{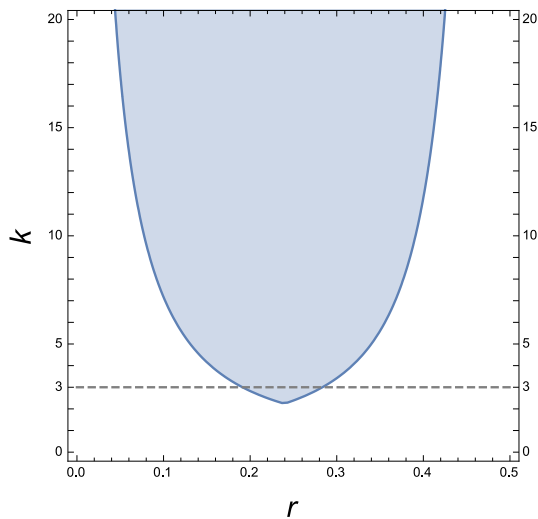}
\caption{The region where $k\geq K_s(r)$ is indicated as the blue region. This region gives how large $k$ is sufficient for each $\omega$ corresponding to distribution $(r,1-r)$ so that $\omega^{\otimes k}$ is a universal resource for asymptotic reversibility. Especially, we find that $k=3$ is sufficient for $r=1/4$.}
\label{figure21}
\end{figure}
\begin{figure}[!t]
\centering
 \includegraphics[clip ,width=2.8in]{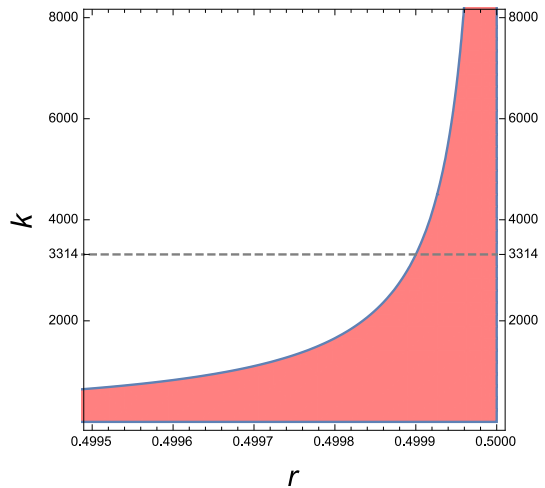}
\caption{The region where $k< K_0(r)$ is indicated as the red region. For each $\omega$ corresponding to distribution $(r,1-r)$, if $(r,k)$ is in this red region, $\omega^{\otimes k}$ is never a universal resource for asymptotic reversibility. In other words, $(r,k)$ must be in the white region at least. Especially, we find that $k=3313$ is needed at least for $r=0.4999$.}
\label{figure22}
\end{figure}

Indeed, a sufficient size $K_s$ and a necessary size $K_0$ for $K$ are given as concrete functions of $\omega$ in Proposition \ref{Ksuf} and Proposition \ref{Kneclem} in Appendix \ref{Proof_assist}, respectively.
Their respective numerical calculations are plotted in Fig. \ref{figure21} and \ref{figure22}.
For example, when the state $\omega_1$ has the squared Schmidt coefficients $(1/4,3/4)$,
$K=3$ is sufficient to be universal due to the calculation of $K_s(1/4)$ as indicated in Fig. \ref{figure21}.
As contrast, when the state $\omega_2$ has the squared Schmidt coefficients $(0.4999,0.5001)$, which is close to the maximally entangled,
$K\geq 3314$ is needed at least due to the calculation of $K_0(0.4999)$ as indicated in Fig. \ref{figure22}.
Thus, interestingly, such a relatively weakly entangled state $\omega_1$ is much better as a supplemental resource for the asymptotic reversibility than a strongly entangled state $\omega_2$ close to the maximally entangled state.
That is, when we use a pure entangled state as a supplemental resource,
the required property is opposite to that for a conventional resource.

\section{Conclusion}\label{conclusion-sec}
\subsection{Summary}
We have addressed the asymptotic compatibility between LOCC conversion and recovery
between two arbitrary bipartite pure entangled states.
We have introduced the MCRE in terms of the Bures distance in order to evaluate the errors of conversion and recovery operations simultaneously, 
and derived the necessary and sufficient condition for their asymptotic compatibleness in Theorem \ref{noloss}. 
Consequently, we have found that LOCC conversion is asymptotically reversible if and only if the conversion characteristics satisfies $\cpq = 1$.
This result suggests $S_\psi/V_\psi=S_\phi/V_\phi$ as a new
kind of asymptotic equivalence relation between pure states.

Moreover, except for lattice case, local unitary operation is found to be enough to realize asymptotically perfect conversion when $\cpq = 1$, and the asymptotic error is small if $\cpq$ is close enough to $1$.
Thus, multiple copies of an arbitrary pure entangled state $\psi$ can be asymptotically stored by LU conversion as the optimal number of copies of another pure entangled state $\phi$ as long as $\cpq = 1$.
Notably, there exists a non-trivial example of the case of $\cpq = 1$ indeed, which means the following interesting fact.
Even when any number of copies of the original state can never be locally unitarily equivalent with any
finite numbers of copies of the target state, 
copies of the original state can be asymptotically converted to
copies of the target state with local unitary operation in the suitable conversion rate,
which is quite different from the situation with finite
$n$.

Then, we have established two methods to realize the asymptotic reversibility of LOCC when it is impossible originally, i.e.,
when $\cpq \neq 1$.
As for the first way, we have considered the strategy to permit some loss in the recovery process.
As a result, it turns out that LOCC conversion and recovery become asymptotically compatible
if and only if the order of the amount of loss is strictly
greater than the square root of the initial number of copies of $\psi$,
as is mentioned in Theorem \ref{HT3}.
Remember that the conventional view \cite{jonathan99:_minim_condit_local_pure_state_entan_manip}\cite{vidal01:_irrev_asymp_manip_entan}\cite{yang05:_irrev_all_bound_entan_states}\cite{acin03}\cite{bennett00:_exact} expects that LOCC conversion is always asymptotically reversible.
Our result rigorously correct this conventional view by the following modification;
When the order $\gamma$ of allowed loss is larger than $1/2$ in the recovered number, LOCC conversion is always asymptotically reversible.
However, under the condition $\cpq=1$, we do not need such a modification because the asymptotic reversibility holds without any loss.
However, this first method does not work for LU conversion, while the following second method does.

The second way is to add an appropriate supplementary state.
The general strategy is to consider the conversion $\psi\otimes\psi'\to \phi$ so that $C_{\psi\otimes\psi', \phi}$ becomes $1$.
In this way, when $S_\psi/V_\psi > S_\phi/V_\phi$,
the ratio $\frac{S_{\psi'}}{V_{\psi'}}$ must be smaller than the ratio $\frac{S_\phi}{V_\phi}$.
Such a state $\psi'$ is further from the maximally entangled state, which
is usually thought to be useless because it has a weaker entanglement.
This fact brings new insight for the resource conversion theory, a utility of far-from-maximally entangled states as a resource.
Although a state $\psi'$
depends on the initial state $\psi$ and the target state $\phi$ in this application,
by modifying this application, any pure entangled state can be universally used for a supplemental resource 
for state conversion in the qubit system
as we have established in Theorem \ref{assist}.
In this modified scheme, we convert $\psi^{\otimes n}$ with $nk$ copies of a supplementary pure entangled state $\omega$
and obtain the target $\phi^{\otimes m_n}$ together with a byproduct state $\omega'^{\otimes m_nk}$ so that this conversion becomes asymptotically reversible.
Theorem \ref{assist} states that for
any non-maximally pure entangled state $\omega$ of the qubit system,
$\omega^{\otimes k}$ can be used as a supplemental resource for asymptotic conversion 
for any two states $\psi$ and $\phi$
when the number of copies $k\geq K$ sufficiently large.
Since sufficient size $K$ can be taken independently of both states $\psi$ and $\phi$,
the state $\omega^{\otimes k}$ can be regarded as a universal resource.
From concrete evaluations of a sufficient size $K_s$ and a necessary size $K_0$ given in Appendix \ref{Proof_assist}, 
we have found that 
the partially entangled state
$\omega$ with squared Schmidt coefficients $(1/4,3/4)$ is much better than nearly maximally entangled states 
as a supplemental resource,
as indicated in Figs. \ref{figure21} and \ref{figure22}.
This fact also indicates a usefulness of far-from-maximally entangled states as a resource.
It is still open to verify whether or not this universality property of any non-maximally entangled states holds in general
dimensional systems.
\subsection{Future Works}
It remains a future problem to exactly formulate the trade-off relation 
between the conversion and the recoverability
for general LOCC conversions 
because we have solved it only for LU conversion.
Based on our analysis, it is not difficult to show that
the error can be improved when the order of the amount of loss is
in the square root of the initial number of copies.
However, it is quite difficult to derive the amount of this improvement.
For entanglement concentration and dilution, this problem has been already solved because the problem is reduced to the individual optimizations of entanglement concentration and dilution \cite{kumagai13:_entan_concen_irrev}.
 However, when the initial or the target state is not 
maximally entangled, 
the problem is not so easy because it seems necessary to consider the optimization of the sum of the errors involving common possible LOCC operations.

The result about LU conversion is useful when we need perfect reversibility.
For example, we can hide entangled state behind the environment so as to secret message encoded in shared entangled state.
Applications of perfect reversibility including this example are expected as future works.

We should note that 
the treatment of the lattice case has a peculiar difficulty
and the proof method for the non-lattice case cannot be directly applied for the lattice case.
In particular, 
although the limit formula (\ref{limfid}) of the fidelity was obtained in the non-lattice case,
it is open whether a similar result holds for the lattice case.
In the lattice case, 
there is room for possibility that there exists an asymptotically reversible conversion with a classical communication, differently from the finite setting.

Moreover, it is important to analyze the LOCC conversion and its reversibility in a finite-length setting for utility, though only asymptotic analysis is treated in this paper.
 For entanglement concentration and dilution, the analysis has been also done and its numerical result has been derived in \cite{kumagai13:_entan_concen_irrev}, but for the general case, it seems to require more ingenuity on account of the same reason.
In fact, even if $\cpq =1$,
the minimum sum of both errors under general LOCC conversion is not zero with finite $n$,
although it should be smaller than the minimum sum under LU conversion.
Since the limit has been shown to be zero under both of them,
we are interested in the convergence speed of the minimum sum.
However, it is also an open problem to clarify 
the asymptotic behavior.

Although our problem is essentially described in terms of probability distributions, this problem for probability distributions is quite far from traditional information theory.
Hence, the conventional method in information theory cannot be applied.
To resolve this problem, we employed two methods.
One is the method of Rayleigh-Normal distribution, which was recently invented in the paper \cite{w.13:_asymp_class_locc_conver_its}.
This method has been used for Theorem \ref{noloss}.
The other is the strong large deviation by Bahadur-Rao \cite{bahadur60:_deviat_sampl_mean}, which brings us a more detailed analysis than conventional large deviation method in information theory.
Indeed, to show Theorem \ref{main}, we need to treat the small difference between two probabilities that are close to each other.
Such a subtle difference has never been addressed in information theory as well as in quantum information theory.
To handle such a subtle difference, we need such a higher order analysis than conventional large deviation.
Maybe, there are some difficult open problems that require such a detailed analysis.
Hence, we can expect that our method can be more widely applied for related areas, e.g., classical information theory, statistical physics.

\begin{acknowledgments}
WK acknowledges support from Grant-in-Aid for JSPS
Fellows No. 233283. MH is partially supported by a MEXT
Grant-in-Aid for Scientific Research (A) No. 23246071 and the
National Institute of Information and Communication Technology
(NICT), Japan. 
The Centre for Quantum Technologies is funded by the
Singapore Ministry of Education and the National Research Foundation
as part of the Research Centres of Excellence programme.

\end{acknowledgments}

\appendix

\section{Majorization and LOCC}\label{defnotes}
We denote by $P_\psi$ the probability distribution composed of the squared Schmidt coefficient of
$\psi \in \hH_A\otimes\hH_B$.
The fidelity between two probability distributions $P$ and $Q$ 
over the same discrete set $\mathcal Y$ 
is defined
as
$
 F(P,Q) := \sum_{y \in \mathcal Y}\sqrt{P(y)}\sqrt{Q(y)}.
$
When the probability distributions $P$ and $Q$ over finite sets satisfy the following, we say that $P$ is majorized by $Q$ and write $P \prec Q$:
\begin{equation}
 \sum_{i=1}^{l}P^\downarrow(i) \leq \sum_{i=1}^{l}Q^\downarrow(i),
\end{equation}
for any $l$, where $P^\downarrow$ is the probability distribution over 
$\natn_{|\mathcal Y|} = \{1,\dots, |\mathcal Y|\}$ 
such that $P^\downarrow(i)$ is the $i$-th largest element of $P$.
Nielsen's theorem states that $\psi$ is transformed to $\psi'$ by LOCC with probability $1$ if and only if $P_\psi\prec P_{\psi'}$ \cite{nielsen99:_condit_class_entan_trans}.
From this theorem we have the following lemma by making use of the concavity of the fidelity.
\begin{Lem}[Vidal et al. \cite{vidal00:_approx}]\label{vidalem}
Let $\psi$ and $\phi$ be pure bipartite entangled states with Schmidt coefficients $\sqrt{P_\psi(j)}$ and $\sqrt{P_\phi(j)}$.
The following relations hold.
  \begin{equation}
\max_{U_A, U_B : \text{unitary}}F((U_A\otimes U_B)\psi, \phi) = F(P_\psi^\downarrow, P_\phi^\downarrow)
  \end{equation}
and
\begin{equation}
 \max_{\Gamma:\text{LOCC}}F(\Gamma(\psi),\phi) = \max_{P_\psi\prec P'}F(P',P_\phi).
\end{equation} 
\end{Lem}
This lemma implies that LOCC $\Gamma$ is better if $\Gamma(\psi)$ is pure and the Schmidt basis of $\Gamma(\psi)$ is the same as that of $\phi$.

\section{Derivation of the Necessary and Sufficient Condition for the Asymptotic Reversibility of LOCC}\label{sa2}
In this appendix,
we show the Necessary and Sufficient Condition for the Asymptotic Reversibility of LOCC given in Theorem \ref{noloss}.
\subsection{Incompatibility when $\cpq\neq 1$}\label{onlyif}
As for a proof of ``only if'' part of Theorem \ref{noloss}, we prove that $\lim_{n\rightarrow\infty}\delta_n(\psi,\phi)\neq 0$ if $\cpq\neq 1$.
It is well known that deterministic LOCC convertibility from a pure state to another pure state is characterized by their majorization order \cite{nielsen99:_condit_class_entan_trans}.
Majorization is a preorder which compares randomness characterized by the Schmidt coefficients (see Appendix \ref{defnotes} for the detail).
Thus it is enough to see only their Schmidt coefficients, 
which reduces our analysis to conversion of probability distributions, 
as indicated by Lemma \ref{vidalem} in Appendix \ref{defnotes} in the following way.
\begin{equation}
 \min_{C:\text{LOCC}}B(C(\psi),\phi) =
 \sqrt{1-\max_{P_\psi\prec P'}F(P',P_\phi)}.\label{Buresfid}
\end{equation}
Moreover, applying the triangle inequality and the monotonicity of the Bures distance,
we find that it is sufficient to deal with LOCC conversion from $\phi^{\otimes m}$ to $\psi^{\otimes n}$ instead of the true recovery LOCC conversion from $C(\psi^{\otimes n})$ to $\psi^{\otimes n}$ as follows.
We have
\begin{align}
& \min_{C,D:\text{LOCC}}B(C(\psi^{\otimes n}),\phi^{\otimes m})+B(\psi^{\otimes n},D\circ C(\psi^{\otimes n}))\nonumber\\
\geq& \min_{C,D:\text{LOCC}}B(D\circ C(\psi^{\otimes n}),D(\phi^{\otimes m}))+\nonumber\\
&\hspace{140pt} B(\psi^{\otimes n},D\circ C(\psi^{\otimes n}))\nonumber\\
\geq& \min_{D:\text{LOCC}}B(\psi^{\otimes n},D(\phi^{\otimes m})),\label{Bureslow}
\end{align}
which implies that
\begin{align}
 &\delta_n(\psi,\phi)\nonumber\\
\geq& \min_{m\in\natn}\max\left\{\min_{C:\text{LOCC}}B(C\right.(\psi^{\otimes n}),\phi^{\otimes m}),\nonumber\\
&\hspace{100pt}\left.\min_{D:\text{LOCC}}B(\psi^{\otimes n},D(\phi^{\otimes m}))\right\}.\label{MCRElow}
\end{align}
Therefore, (\ref{Buresfid}) implies that it is enough to show that
 \begin{align}
 \lim_{n\rightarrow\infty}\max_{m\in\natn}\min\{\max_{P_\psi^n\prec P'}F(P',P_\phi^{m}),\max_{P_\phi^m\prec P''}F(P'',P_\psi^n)\}\nonumber\\
 < 1,\label{fidmax}
 \end{align}
when $\cpq \neq 1$. To prove this, we focus on the maximizer $m_n$ with respect to $m$ in the left hand side of (\ref{fidmax}),
and expand it as $m_n = S_\psi/S_\phi n + b n^{\gamma} + o(n^{\gamma})$.

When $\gamma>1/2$, Theorem 10 in \cite{w.13:_asymp_class_locc_conver_its}
implies that
\begin{align}
 \lim_{n\rightarrow\infty}\max_{P_\psi^n\prec P'}F(P',P_\phi^{S_\psi/S_\phi n + b n^{\gamma} + o(n^{\gamma})}) &=& 0 \quad\left(b>0\right),\nonumber\\
&&\label{ag}\\
\lim_{n\rightarrow\infty}\max_{P_\phi^{S_\psi/S_\phi n + b n^{\gamma} + o(n^{\gamma})}\prec P''}F(P'',P_\psi^n)&=&0\quad\left(b < 0\right).\nonumber\\
&&\label{al}
\end{align}
These mean that when $b<0$ and $\gamma>1/2$, the conversion is possible since the target number of copies is small, however the recovery is no longer impossible because of fatal lost of number of copies, and vice versa.
Hence, (\ref{fidmax}) holds if $b \neq 0$ and $\gamma>1/2$.

On the other hand, the case when $\gamma=1/2$ requires more delicate analysis because 
the reversibility is just determined by $\cpq$.
It is so-called second order analysis.
We apply Theorem 10 and Lemma 11 in \cite{w.13:_asymp_class_locc_conver_its} to the case of $\cpq \neq 1$.
Then, using the Rayleigh-Normal distribution $Z_{v}(\mu)$ \cite{w.13:_asymp_class_locc_conver_its}, we can describe the 
limit of the maximum of the fidelity as
\begin{align}
&\lim_{n\rightarrow\infty}\max_{P_\psi^n\prec P'}F(P',P_\phi^{S_\psi/S_\phi n + b\sqrt{n} + o(\sqrt{n})}) \nonumber\\
&= \sqrt{1-Z_{\cpq}\left(b\frac{S_\phi}{\sqrt{V_\psi}}\right)}<1 \label{Fpq}
\end{align}
for any $b\in\real$.
Overall, we have shown (\ref{fidmax}) when $\cpq\neq 1$.

\subsection{Achievability when $\cpq=1$}\label{proof-main-achieve}
Now, we construct achievable conversions with $m_n=S_\psi/S_\phi n + o(\sqrt{n})$ when $\cpq=1$ by using the deterministic conversion as follows.
For a map $W : \mathcal{X} \rightarrow \mathcal{Y}$,
we define the deterministic conversion map from 
a probability distribution over $\mathcal{X}$
to another probability distribution over $\mathcal{Y}$ by using the same symbol $W$ as follows.
\begin{equation}
W(P)(y) := \sum_{x \in W^{-1}(y)}P(x).
\end{equation}
The following lemma implies the possibility of the desired conversions.
\begin{Lem}[Kumagai and Hayashi \cite{w.13:_asymp_class_locc_conver_its}]
For probability distributions $P_1$ over $\mathcal{X}$ and $P_2$ over $\mathcal{Y}$, we have
 \begin{align}
  \lim_{n\rightarrow\infty}\max\{F(W(P_1^n),P_2^{m_n})|W:\mathcal{X}\rightarrow \mathcal{Y}\}=1
 \end{align}
when $C_{P_1,P_2}=1$ and $m_n=S(P_1)n/S(P_2)+o(\sqrt{n})$,
where $S(P_i)$ denotes the Shannon entropy of $P_i$, and $C_{P_1,P_2}:=(S(P_1)/V(P_1))(S(P_2)/V(P_2))^{-1}$ where $V(P_i)=\sum_x P_i(x)(-\log P_i(x)-S(P_i))^2$.
\end{Lem}
From this lemma, there exist two sequences of deterministic conversions $(W_n)_{n\in\natn}$ and $(X_n)_{n\in\natn}$ such that
\begin{align}
 \lim_{n\rightarrow\infty} F(W_n(P_\psi^{n}), P_\phi^{\frac{S_\psi}{S_\phi}n}) &=  \lim_{n\rightarrow\infty}F(X_n(P_\phi^{\frac{S_\psi}{S_\phi}n}),P_\psi^{n}) \nonumber\\
&= 1,\label{determ1}
\end{align}
when $\cpq = 1$.
Since 
$P_\psi^{n} \prec W_n(P_\psi^{n})$
and 
$ P_\phi^{\frac{S_\psi}{S_\phi}n} \prec
X_n(P_\phi^{\frac{S_\psi}{S_\phi}n})$,
there exist two sequences of LOCC operations $(C_n)_{n\in\natn}$ and $(D_n)_{n\in\natn}$ such that
\begin{align}
 B(\phi^{\otimes \frac{S_\psi}{S_\phi}n},C_n(\psi^{\otimes n}))^2 &= 1 - F(W_n(P_\psi^{n}), P_\phi^{\frac{S_\psi}{S_\phi}n}),\\
 B(\psi^{\otimes n},D_n(\phi^{\otimes \frac{S_\psi}{S_\phi}n}))^2 &= 1 - F(X_n(P_\phi^{\frac{S_\psi}{S_\phi}n}),P_\psi^{n}).
\end{align}
 Then, it follows from the triangle inequality and the monotonicity of the Bures distance again that
\begin{align}
& B(\psi^{\otimes n},D_n\circ C_n(\psi^{\otimes n}))\nonumber\\
\leq & B(\psi^{\otimes n},D_n(\phi^{\otimes \frac{S_\psi}{S_\phi}n})) + B(D_n(\phi^{\otimes \frac{S_\psi}{S_\phi}n}),D_n\circ C_n(\psi^{\otimes n}))\nonumber\\
\leq & B(\psi^{\otimes n},D_n(\phi^{\otimes \frac{S_\psi}{S_\phi}n})) + B(\phi^{\otimes \frac{S_\psi}{S_\phi}n},C_n(\psi^{\otimes n}))\nonumber\\
=& \sqrt{1-F(X_n(P_\phi^{\frac{S_\psi}{S_\phi}n}),P_\psi^{n})}+\sqrt{1-F(W_n(P_\psi^{n}), P_\phi^{\frac{S_\psi}{S_\phi}n})}\nonumber\\
&\rightarrow 0\qquad (n\rightarrow\infty).\label{BuresUp}
\end{align}
Thus,
\begin{align}
&\lim_{n\rightarrow\infty}\delta_n(\psi,\phi)\nonumber\\
\leq& \lim_{n\rightarrow\infty}B(\phi^{\otimes \frac{S_\psi}{S_\phi}n},C_n(\psi^{\otimes n})) + B(\psi^{\otimes n},D_n\circ C_n(\psi^{\otimes n}))\nonumber\\
=& 0.
\end{align}

\section{Detailed Asymptotic Analysis of the Fidelity concerning the LU Conversion}\label{proofmain}
Firstly, we show Theorem \ref{low-err} in Appendix \ref{sa31}.
Latter, we show Lemma \ref{mainlem} in Appendix \ref{sa32}.

\subsection{An Upper Bound of the Fidelity}\label{sa31}
To begin with, we establish the following upper bound of the fidelity which is valid even for the lattice distributions, which directly proves Theorem \ref{low-err} in Sec. \ref{LU-main}.
\begin{Lem}\label{upper-thm}
 The following holds.
\begin{align}
 &\lim_{n\rightarrow\infty}F(P_{\psi}^{n\downarrow},P_{\phi}^{m_n\downarrow})\nonumber\\
\leq&\left\{
\begin{array}{ll}
 \displaystyle
0\quad &\text{if } \frac{\left|m_n-\frac{S_{\psi}}{S_{\phi}}n\right|}{\sqrt{n}}\rightarrow\infty\\
\sqrt{\frac{2}{\cpq^{\frac{1}{2}}+\cpq^{-\frac{1}{2}}}}\exp \left[-\frac{b^2S_{\phi}^2}{4(1+\cpq)V_{\psi}}\right]\quad &\text{if  } \frac{m_n-\frac{S_{\psi}}{S_{\phi}}n}{\sqrt{n}}\rightarrow b.
\end{array}
\right.
\end{align}
\end{Lem}
 Because of the positivity of the fidelity, Lemma \ref{upper-thm} immediately proves that $\lim_{n\rightarrow\infty}F(P_\psi^{n\downarrow},P_\phi^{m_n\downarrow})= 0$ for $m_n$ such that $|m_n-\frac{S_\psi}{S_\phi}n|/\sqrt{n}\rightarrow \infty$.

Now, we prove this upper bound.
We define $S_n^{\omega}(x):=\{1,2,\dots,\lceil e^{S_\omega n+x\sqrt{n}\rceil}\}$ and $S_n^{\omega}(x,x'):=S_n^{\omega}(x')\backslash S_n^{\omega}(x)$ for $\omega=\psi, \phi$.
Now we prepare a lemma.
\begin{Lem}
 \begin{align}
  &\lim_{n\rightarrow \infty}P_\psi^{n\downarrow}(S_n^{\psi}(x))=\Phi(\frac{x}{\sqrt{V_\psi}})\\
  &\lim_{n\rightarrow \infty}P_\phi^{m_n\downarrow}(S_n^{\psi}(x))\nonumber\\
  =&\left\{
\begin{array}{ll}
 \displaystyle
0\quad & \text{if } \frac{m_n-\frac{S_\psi}{S_\phi}n}{\sqrt{n}}\rightarrow \infty \\
1\quad & \text{if } \frac{m_n-\frac{S_\psi}{S_\phi}n}{\sqrt{n}}\rightarrow -\infty \\
\Phi_{\psi,\phi,b}\left(\frac{x}{\sqrt{V_\psi}}\right)\quad 
& \text{if }\frac{m_n-\frac{S_\psi}{S_\phi}n}{\sqrt{n}}\rightarrow b
\end{array}
\right.
 \end{align}
holds for an arbitrary real number $x$, where $\Phi_{\psi,\phi,b}(x):=\Phi\left(\sqrt{\cpq^{-1}}(x-\frac{bS_\phi}{\sqrt{V_\psi}})\right)$.
\end{Lem}
This is immediately followed from Lemma 12 in \cite{w.13:_asymp_class_locc_conver_its}.
For an arbitrary fixed small real number $\epsilon >0$,
we choose a sufficiently large real number $R$ to satisfy
\begin{align}
 \sqrt{\Phi(-R)}\sqrt{\Phi_{\psi,\phi,b}(-R)}+\sqrt{1-\Phi(R)}\sqrt{1-\Phi_{\psi,\phi,b}(R)}<\epsilon.
\end{align}
Moreover, we introduce $c_j^N:=-R+\frac{2R}{N}j$ and $\tilde{c_j}^N:=\sqrt{V_\psi}c_j^N$.
Then we have
\begin{align}
 &F(P_\psi^{n\downarrow},P_\phi^{m_n\downarrow})\nonumber\\
\leq&\sqrt{P_\psi^{n\downarrow}(S_n^{\psi}(\tilde{c_0}^N))P_\phi^{m_n\downarrow}(S_n^{\psi}(\tilde{c_0}^N))}\nonumber\\
&+\sum_{j=1}^N\sqrt{P_\psi^{n\downarrow}(S_n^{\psi}(\tilde{c_{j-1}}^N,\tilde{c_j}^N))P_\phi^{m_n\downarrow}(S_n^{\psi}(\tilde{c_{j-1}}^N,\tilde{c_j}^N))}\nonumber\\
&+\sqrt{1-P_\psi^{n\downarrow}(S_n^{\psi}(\tilde{c_N}^N))}\sqrt{1-P_\phi^{m_n\downarrow}(S_n^{\psi}(\tilde{c_N}^N)}.\label{finiteup}
\end{align}
Taking the limit $n\rightarrow\infty$, the right hand side of (\ref{finiteup}) becomes $\sqrt{\Phi(-R)}<\epsilon$ ($\sqrt{1-\Phi(R)}<\epsilon$) when $\frac{m_n-\frac{S_\psi}{S_\phi}n}{\sqrt{n}}\rightarrow -\infty$ ($\frac{m_n-\frac{S_\psi}{S_\phi}n}{\sqrt{n}}\rightarrow \infty$), respectively.
Since this is true for any $\epsilon >0$, $\lim_{n\rightarrow\infty}F(P_\psi^{n\downarrow},P_\phi^{m_n\downarrow})\leq 0$ holds.
When $(m_n-\frac{S_\psi}{S_\phi}n)/\sqrt{n}\rightarrow b$, taking the limit $N\rightarrow\infty$ after taking $n\rightarrow\infty$, the right hand side of (\ref{finiteup}) becomes
\begin{align}
&\lim_{N\rightarrow\infty}
\sum_{j=1}^N
\sqrt{\Phi(c_j^N)-\Phi(c_{j-1}^N)}
 \sqrt{\Phi_{\psi,\phi,b}(c_j^N)-\Phi_{\psi,\phi,b}(c_{j-1}^N)}\nonumber\\
 &+\sqrt{\Phi(-R)}\sqrt{\Phi_{\psi,\phi,b}(-R)}+\sqrt{1-\Phi(R)}\sqrt{1-\Phi_{\psi,\phi,b}(R)}\nonumber\\
=&\int_{-R}^R\sqrt{\od{\Phi}{t}(t)}\sqrt{\od{\Phi_{\psi,\phi,b}}{t}(t)}d t\nonumber\\
 &+\sqrt{\Phi(-R)}\sqrt{\Phi_{\psi,\phi,b}(-R)}+\sqrt{1-\Phi(R)}\sqrt{1-\Phi_{\psi,\phi,b}(R)}\nonumber\\
<&\int_{-R}^R\sqrt{\od{\Phi}{t}(t)}\sqrt{\od{\Phi_{\psi,\phi,b}}{t}(t)}d t+\epsilon\nonumber\\
\leq&\int_{-\infty}^{\infty}\sqrt{\od{\Phi}{t}(t)}\sqrt{\od{\Phi_{\psi,\phi,b}}{t}(t)}d t+\epsilon\nonumber\\
=& \sqrt{\frac{2}{\cpq^{\frac{1}{2}}+\cpq^{-\frac{1}{2}}}}\exp \left[-\frac{b^2S_\phi^2}{4(1+\cpq)V_\psi}\right] + \epsilon.
\label{fidupp}
\end{align}
Since $\epsilon>0$ is arbitrary, (\ref{fidupp}) implies that
\begin{align}
 &\lim_{n\rightarrow\infty}F(P_\psi^{n\downarrow},P_\phi^{m_n\downarrow})\nonumber\\
\leq&
\sqrt{\frac{2}{\cpq^{\frac{1}{2}}+\cpq^{-\frac{1}{2}}}}\exp \left[-\frac{b^2S_\phi^2}{4(1+\cpq)V_\psi}\right]
\end{align}
when $\frac{m_n-\frac{S_\psi}{S_\phi}n}{\sqrt{n}}\rightarrow b$ is true.

\subsection{The Limit of the Fidelity for Non-Lattice Distributions}\label{sa32}
Now, we establish Lemma \ref{mainlem} in Sec. \ref{LU-main}, i.e., 
the formula of the limit of the fidelity concerning the LU conversion.
Since we have already proved that $\lim_{n\rightarrow\infty}F(P_\psi^{n\downarrow},P_\phi^{m_n\downarrow})=0$ when $\left|m_n-\frac{S_{\psi}}{S_{\phi}}n\right|/\sqrt{n}\rightarrow\infty$ in Lemma \ref{upper-thm},
for the proof of Lemma \ref{mainlem},
it is enough to show (\ref{limfid}) only when
neither $P_\psi$ nor $P_\phi$ is a lattice distribution and $(m_n-\frac{S_\psi}{S_\phi}n)/\sqrt{n}\rightarrow b$.

We introduce the cumulant generating functions
$g_\omega(1+s):= \log \tr (\tr_B\omega)^{1+s}$ for $\omega=\psi,\phi$.
They satisfy
\begin{align}
 g_\omega'(1)&=-S_\omega\\
 g_\omega''(1)&=V_\omega
\end{align}
Since $g_\omega$ is strictly convex, 
we can define the inverse function of $g_\omega'$, which is denoted by $h_\omega$.
Then, we have
\begin{align}
 h_\omega'(-S_\omega)&=\frac{1}{V_\omega}
\end{align}
Now, we define the random variable 
$Z_{\omega,n}:= (\log P^{n,\downarrow}_{\omega}(\hat{j})+ n S_\omega)/ \sqrt{n}$.
Then, when 
$j= P_C \{ Z_{\omega,n} \ge a  \}$,
$j$ is the maximum integer satisfying 
$\log P^{n,\downarrow}_{\omega}(j) \ge - n S_\omega+ \sqrt{n} a$, where $P_C$ is the counting measure. We focus on the Legendre transform of $g_\omega$, which is written as 
\begin{align}
\max_s s R - g_\omega(s)
=  h_\omega(R) R - g_\omega(h_\omega(R)).
\end{align}
We employ the strong large deviation by Bahadur and Rao \cite{bahadur60:_deviat_sampl_mean}\cite{joutard13:_stron_large_deviat_theor}:
\begin{align}
&\log P_C \{ \log P^{n,\downarrow}_{\omega} (\hat{j}) \ge n R  \} \nonumber\\
=& 
-n (h_\omega(R) R - g_\omega(h_\omega(R)))\nonumber\\
&- \log (\sqrt{2 \pi n g_\omega''(h_\omega(R))} h_\omega(R))+o(1)\nonumber\\
=& 
-n (h_\omega(R) R - g_\omega(h_\omega(R)))
- \log \sqrt{2 \pi n }
- \log h_\omega(R)\nonumber\\
&+ \frac{1}{2}\log {h_\omega}'(R)+o(1) .
\label{6-1-2}
\end{align}
In the following, for a unified treatment, 
the functions $- (h_\omega(R) R - g_\omega(h_\omega(R)))$ and 
$- \log \sqrt{2 \pi }
- \log h_\omega(R)
+ \frac{1}{2}\log {h_\omega}'(R)) $ are written as
$f_{\omega,0} (R)$ and $f_{\omega,1} (R)$.
So, (\ref{6-1-2}) is simplified as
$\log P_C \{ \log p^{n,\downarrow}_{\omega} (\hat{j}) \ge n R  \} 
= \sum_{k=0}^{1} f_{\omega,k} (R) n^{1-k} -\frac{1}{2} \log n+o(1) $. 

Thus, we have
\begin{align}
&
\log J_{\omega,n}(a)
:=\log P_C \{ \log P^{n,\downarrow}_{\omega} (\hat{j}) \ge - n S+ \sqrt{n}a \}
\\
=& \sum_{k=0}^{1} \sum_{t=0}^{2(1-k)}
\frac{f_{\omega,k}^{(t)} (-S_\omega) n^{1-k-\frac{t}{2}}  }{t !} a^t
-\frac{1}{2} \log n+o(1) . 
\end{align}

To calculate the fidelity, we define 
$\Delta Z$ as
\begin{align}
\log J_{\phi,m_n}\left(\sqrt{\frac{n}{m_n}}(Z_{\psi,n}+\Delta Z)\right)-\log J_{\psi,n}(Z_{\psi,n})=o(1). \label{6-6-1}
\end{align}
Define 
$\alpha_{\psi,0}(Z_{\psi,n})$ and $\alpha_{\phi,i}(Z_{\psi,n})$ as
\begin{align}
&\alpha_{\psi,0}(Z_{\psi,n})\nonumber\\
:=&\sum_{k=0}^{1}
\sum_{j=0}^{2(1-k)}
n^{1-k-\frac{j}{2}} \frac{Z_{\psi,n}^j}{j !} f_{\psi,k}^{(j)}(-S_\psi) ,
\\&\nonumber\\
&\alpha_{\phi,i}(Z_{\psi,n})\nonumber\\
:=&\sum_{k=0}^{1-\frac{i}{2}}
\sum_{j=0}^{2(1-k)-i}
m_n^{1-k-(i+j)}n^{\frac{i+j}{2}} \frac{Z_{\psi,n}^j}{j !} f_{\phi,k}^{(i+j)}(-S_\phi) ,
\end{align}
where $f_{\omega,k}^{(i)} $ is the $i$-th derivative of $f_{\omega,k}$.
Then,
$\Delta Z$ 
satisfies the equation
\begin{align}
&\sum_{i=1}^{2} 
\alpha_{\phi,i}(Z_{\psi,n}) \frac{(\Delta Z)^i}{i !} \nonumber\\
=& -(\alpha_{\phi,0}(Z_{\psi,n})-\alpha_{\psi,0}(Z_{\psi,n})-\frac{1}{2}\log\frac{m_n}{n})+o(1).\nonumber\\
\label{6-5-2}
\end{align}
Due to the definition (\ref{6-6-1}), we have
\begin{align}
&F(P_\psi^{n\downarrow},P_\phi^{m_n\downarrow})\nonumber\\
=&
\sum_{j} P^{n\downarrow}_{\psi} (j)
\sqrt{\frac{P^{m_n\downarrow}_{\phi} (j)}{P^{n\downarrow}_{\psi} (j)}}
\nonumber\\
=&
\sum_{j} P^{n\downarrow}_{\psi} (J_{\psi,n}(Z_{\psi,n}(j)))
\sqrt{\frac{P^{m_n\downarrow}_{\phi} (J_{\psi,n}(Z_{\psi,n}(j)))}{P^{n\downarrow}_{\psi} (J_{\psi,n}(Z_{\psi,n}(j)))}}
\nonumber\\
=&E
\Bigg[
e^{\frac{1}{2}\left[\log P^{m_n\downarrow}_{\phi} (J_{\psi,n}(Z_{\psi,n}(\hat{j})))-\log P^{n\downarrow}_{\psi} (J_{\psi,n}(Z_{\psi,n}(\hat{j})))\right]}
\Bigg]
\nonumber\\
=&
E
\Bigg[
e^{\frac{1}{2}\left[-m_n S_\phi+\sqrt{m_n}\sqrt{\frac{n}{m_n}}(Z_{\psi,n}+\Delta Z)-(-n S_\psi + \sqrt{n}Z_{\psi,n})+o(1)\right]}
\Bigg] \nonumber\\
=&
E
\Bigg[
\exp\frac{1}{2}(n S_\psi -m_n S_\phi+ \sqrt{n}\Delta Z+o(1))
\Bigg]
 .\label{6-5-5}
\end{align}
We note that the fourth equality of (\ref{6-5-5}) can be established since both $P_\psi$ and $P_\phi$ are non-lattice distribution, but it is not the case when one of them is a lattice distribution.

Now, it is needed to solve the equation \eqref{6-5-2} with respect to 
$\Delta Z$.
Let $m_n=\frac{S_\psi}{S_\phi}n+b\sqrt{n}$ for $b\in\real$.
Notice that $\alpha_{\phi,i}(Z_{\psi,n})=O(n^{1-\frac{i}{2}})$.
We apply Lemma \ref{l6-6} in the next section to the equation \eqref{6-5-2} with
$x= \frac{\Delta Z}{\sqrt{n}}$,
$a_i= \alpha_{\phi,i}(Z_{\psi,n}) n^{\frac{i}{2}-1}/i!$,
and
$\epsilon= (-\alpha_{\phi,0}(Z_{\psi,n})+\alpha_{\psi,0}(Z_{\psi,n})+\frac{1}{2}\log\frac{S_\psi}{S_\phi})/n$.
Then, we obtain
\begin{align}
&\sqrt{n}\Delta Z\nonumber\\
=& bS_\phi\sqrt{n}-\frac{1}{2}\left[\cpq^{-1}(T-\frac{bS_\phi}{\sqrt{V_\psi}})^2-T^2\right]+\frac{1}{2}\log\cpq^{-1}+o(1),
\end{align}
where $T := Z_{\psi,n}/\sqrt{V_\psi}$.
Then, we have
\begin{align}
 &F(P_\psi^{n\downarrow},P_\phi^{m_n\downarrow})\nonumber\\
=&E
\Bigg[
\exp\frac{1}{2}(-bS_\phi\sqrt{n}+ \sqrt{n}\Delta Z+o(1))
\Bigg]\nonumber\\
=&E
\Bigg[
e^{-\frac{1}{4}\left[\cpq^{-1}(T-\frac{bS_\phi}{\sqrt{V_\psi}})^2-T^2\right]+\frac{1}{4}\log\cpq^{-1}+o(1)}
\Bigg].
\label{finitelow}
\end{align}
(\ref{finitelow}) is unchanged even if $b$ is replaced with $b+o(1)$, hence it holds whenever $\frac{m_n-\frac{S_\psi}{S_\phi}n}{\sqrt{n}}\rightarrow b$ holds.
By the central limit theorem, we obtain
\begin{align}
& E
\Bigg[
e^{-\frac{1}{4}\left[\cpq^{-1}(T-\frac{bS_\phi}{\sqrt{V_\psi}})^2-T^2\right]+\frac{1}{4}\log\cpq^{-1}+o(1)}
\Bigg]\nonumber\\
\rightarrow&\int \cpq^{-\frac{1}{4}}e^{-\frac{1}{4}\left[\cpq^{-1}(t-\frac{bS_\phi}{\sqrt{V_\psi}})^2-t^2\right]}d \Phi(t) \nonumber\\
=&\sqrt{\frac{2}{\cpq^{\frac{1}{2}}+\cpq^{-\frac{1}{2}}}}\exp \left[-\frac{b^2S_\phi^2}{4(1+\cpq)V_\psi}\right].
\end{align}
Therefore, we obtain the desired argument.

\subsection{Perturbation for higher order equation}

\begin{Lem}\label{l6-6}
Consider the equation
\begin{align}
\epsilon =\sum_{i=1}^l a_i x^i.\label{6-6-11}
\end{align}
When $\epsilon$ is sufficiently small, 
the solution $x$ is approximated as
\begin{align}
x= \sum_{i=1}^l \epsilon^i x_i + O(\epsilon^{l+1}).\label{6-6-10}
\end{align}
where
$x_1$ is given as $\frac{1}{a_1} $
and 
$x_l$ with $l \ge 2$ is inductively given as
$- \frac{1}{a_1}
\sum_{i_1,i_2,\ldots, i_{l-1}:\sum_{k=1}^{l-1}k i_k=l}
a_{\sum_{k=1}^{l-1}i_k} 
\frac{(\sum_{k=1}^{l-1}i_k)!}{\prod_{k=1}^{l-1} i_k! }
\prod_{k=1}^{l-1}x_k^{i_k}$.
Especially, $x_2$ is given as
\begin{align}
x_2 &= -\frac{a_2}{a_1} x_1^2 =-\frac{a_2}{a_1^3}
\end{align}
\end{Lem}
This lemma can be shown as follows.
First, we substitute \eqref{6-6-10} into \eqref{6-6-11}.
Then, compare the coefficients with the order $\epsilon^i$.
Hence, we obtain $x_l=
- \frac{1}{a_1}
\sum_{i_1,i_2,\ldots, i_{l-1}:\sum_{k=1}^{l-1}k i_k=l}
a_{\sum_{k=1}^{l-1}i_k} 
\frac{(\sum_{k=1}^{l-1}i_k)!}{\prod_{k=1}^{l-1} i_k! }
\prod_{k=1}^{l-1}x_k^{i_k}$.

\section{Detailed Establishment of Theorem \ref{HT3}}\label{Proof_HT3}
\subsection{Achievability}
To show the achievability, 
we show the following lemma.
\begin{Lem}\label{FirstLoss}
The following relation
holds with $1 > \gamma > \frac{1}{2}$ regardless of $\cpq$:
\begin{align}
&\lim_{n\rightarrow \infty}\min_{m\in\natn , C,D:\text{LOCC}} B(C(\psi^{\otimes n}),\phi^{\otimes m})\nonumber\\
&\hspace{100pt}+B(\psi^{\otimes (n-n^{\gamma})},D\circ C(\psi^{\otimes n}))\nonumber\\
=& 0.\label{largeloss}
\end{align}
\end{Lem}
This lemma implies that LOCC conversion is always asymptotically preservable when the order of the amount of loss is greater than the square root of the initial number of copies of $\psi$ even if $\cpq  \neq 1$.
\begin{proof}
At first,
\begin{equation}
\lim_{n\rightarrow\infty}\max_{P_\psi^n\prec P'}
F(P',P_\phi^{\frac{S_{\psi}}{S_{\phi}}(n-\frac{1}{2}n^{\gamma})}) 
= 1 \label{firstloss}
\end{equation}
holds from Theorem 10 in \cite{w.13:_asymp_class_locc_conver_its}.
Now, we define the sequence $a_n$ as
$n-n^{\gamma}=
\frac{S_{\phi}}{S_{\psi}}
[\frac{S_{\psi}}{S_{\phi}}(n-\frac{1}{2}n^{\gamma})
-a_n]$.
Since $a_n =\Omega(
(\frac{S_{\psi}}{S_{\phi}}(n-\frac{1}{2}n^{\gamma}))^{\gamma}
)$,
the following holds similarly to (\ref{firstloss}).
\begin{align}
 &\lim_{n\rightarrow\infty}
\max_{P_\phi^{\frac{S_{\psi}}{S_{\phi}}(n-\frac{1}{2}n^{\gamma})}\prec P''}
F(P'',P_\psi^{(n-n^{\gamma})})\nonumber\\
=& \lim_{n\rightarrow\infty}\max_{P_\phi^{\frac{S_{\psi}}{S_{\phi}}(n-\frac{1}{2}n^{\gamma})}\prec P''}
F(P'',P_\psi^{
\frac{S_{\phi}}{S_{\psi}}
[\frac{S_{\psi}}{S_{\phi}}(n-\frac{1}{2}n^{\gamma})
-a_n]})
\nonumber\\
=&1.
\end{align}
Thus, since (\ref{Buresfid}) implies that
\begin{align}
& \lim_{n\rightarrow\infty} \min_{C:\text{LOCC}}B(\phi^{\otimes \frac{S_{\psi}}{S_{\phi}}(n-\frac{1}{2}n^{\gamma})},C(\psi^{\otimes n})) \nonumber\\
&\qquad +\min_{D:\text{LOCC}}B(\psi^{\otimes (n-n^{\gamma})},D(\phi^{\otimes \frac{S_{\psi}}{S_{\phi}}(n-\frac{1}{2}n^{\gamma})}))\nonumber\\
=& \lim_{n\rightarrow\infty}\sqrt{1-\max_{P_\psi^n\prec P'}F(P',P_\phi^{\frac{S_{\psi}}{S_{\phi}}(n-\frac{1}{2}n^{\gamma})})}\nonumber\\
&\qquad+\sqrt{1-\max_{P_\phi^{\frac{S_{\psi}}{S_{\phi}}(n-\frac{1}{2}n^{\gamma})}\prec P''}F(P'',P_\psi^{(n-n^{\gamma})})}\nonumber\\
=& 0, \nonumber
\end{align}
(\ref{largeloss}) is proven by applying the triangle inequality and the monotonicity of the Bures distance.
\end{proof}

\subsection{Impossibility}
Finally, we see the preservability when the order of the amount of loss is only the square root of the initial number of copies and $\cpq \neq 1$. We investigate the following asymptotic error as before.
\begin{align}
& \delta^{(2)}_{\infty}(\psi,\phi, \beta)\nonumber\\
:=&\lim_{n\rightarrow \infty}\min_{m\in\natn , C,D:\text{LOCC}} B(C(\psi^{\otimes n}),\phi^{\otimes m})\nonumber\\
&\hspace{90pt}+B(\psi^{\otimes n-\beta\sqrt{n}},D\circ C(\psi^{\otimes n})).\nonumber\\
\end{align}
\begin{Lem}\label{SecondLoss}
The inequality $\delta^{(2)}_{\infty}(\psi,\phi, \beta) > 0$
holds for any $\beta\in\real$ when $\cpq \neq 1$.
\end{Lem}

Combining Lemmas \ref{FirstLoss} and \ref{SecondLoss}, 
we obtain Theorem \ref{HT3}.
\begin{proof}
Similarly to (\ref{Bureslow}),
\begin{equation}
 \delta^{(2)}_{\infty}(\psi,\phi, \beta)\geq\lim_{n\rightarrow\infty}\min_{D:\text{LOCC}}B(\psi^{\otimes n-\beta\sqrt{n}},D(\phi^{\otimes m}))
\end{equation}
holds. Thus, the following holds the same as (\ref{MCRElow}).
\begin{align}
 &\delta^{(2)}_{\infty}(\psi,\phi, \beta)\nonumber\\
\geq& \min_{m\in\natn}\max\left\{\min_{C:\text{LOCC}}\right.B(C(\psi^{\otimes n}),\phi^{\otimes m}),\nonumber\\
&\hspace{90pt}\left.\min_{D:\text{LOCC}}B(\psi^{\otimes n-\beta\sqrt{n}},D(\phi^{\otimes m}))\right\}\nonumber\\
=& \left[ 1 - \max_{m\in\natn}\min\right.\left\{\max_{P_\psi^n\prec P'}\right.F(P',P_\phi^{m}),\nonumber\\
&\hspace{90pt}\left.\left.\max_{P_\phi^{m}\prec P''}F(P'',P_\psi^{n-\beta\sqrt{n}})\right\}\right]^{\frac{1}{2}}.\nonumber
\end{align}
Then, 
\begin{equation}
 \lim_{n\rightarrow\infty}\max_{P_\phi^{an + o(n)}\prec P''}F(P'',P_\psi^{n-\beta\sqrt{n}})=0\qquad\left(a < \frac{S_{\psi}}{S_{\phi}}\right) \label{alSecond}
\end{equation}
holds from the first order asymptotics similarly to (\ref{al}). Thus,
(\ref{ag}), (\ref{alSecond}), and (\ref{Fpq})
imply that
\begin{align*}
& \lim_{n\rightarrow\infty}\max_{m\in\natn}\min\left\{\max_{P_\psi^n\prec P'}\right.F(P',P_\phi^{m}),
\nonumber\\&\hspace{90pt}
\left.\max_{P_\phi^{m}\prec P''}F(P'',P_\psi^{n-\beta\sqrt{n}})\right\} 
 < 1
\end{align*}
when $\cpq \neq 1$, similarly to the proof of (\ref{fidmax}). Therefore, the lemma is proven.
\end{proof}

\section{Detailed Discussion of Theorem \ref{assist}}\label{Proof_assist}
The purpose of this appendix is deriving upper and lower bounds for $K$ of Theorem \ref{assist}.
For this purpose, 
we prepare some notations.
Let $h$ and $v$ denote the binary entropy $h(p)=-p\log p-(1-p)\log(1-p)$, and $v(p):=p(\log p)^2+(1-p)(\log (1-p))^2-h(p)^2$, respectively.
Then, we define $C:=\max_{0\leq p\leq 1/2}v(p)$ and $x_r:=g^{-1}(v(r)/h(r)+2)$, where $g(x):=(1-2x)(\log (1-x)-\log x)$.
Note that $g$ is a strictly monotonically decreasing function on $(0,1/2]$, thus $g^{-1}$ is well defined on $\text{ran} g=[0,\infty)$.
Then, we define $K_0(r)$, $K_1(r)$, and $K_2(r)$ as follows.
  \begin{align}
   &K_0(r)\nonumber\\
   :=&\frac{(h(r)-\log 2)v(x_r)-v(r)h(x_r)}{2\log 2v(r)}\nonumber\\
   &\hspace{-13pt}+\hspace{-3pt}\frac{\sqrt{[(\log 2\hspace{-1pt}-\hspace{-1pt}h(r))v(x_r)\hspace{-2pt}+\hspace{-2pt}v(r)h(x_r)]^2\hspace{-2pt}+\hspace{-2pt}4(\log 2)^2v(x_r)v(r)}}{2\log 2v(r)},\label{Knec}\\
  &K_1(r)\nonumber\\
  :=&\frac{Ch(x_r)+v(r)\log 2}{2(h(r)v(x_r)-v(r)h(x_r))}\nonumber\\
  &\hspace{-10pt}+\hspace{-2pt}\frac{\sqrt{\hspace{-1pt}(Ch(x_r)\hspace{-2pt}+\hspace{-2pt}v(r)\log 2)^2\hspace{-2pt}+\hspace{-2pt}4C\hspace{-1.5pt}\log 2(h(r)v(x_r)\hspace{-2pt}-\hspace{-2pt}v(r)h(x_r))}}{2(h(r)v(x_r)-v(r)h(x_r))},
 \end{align}
and
 \begin{align}
  K_2(r):=\frac{Ch(r)+\sqrt{C^2h(r)^2+4(\log 2)^2Cv(r)}}{2v(r)\log 2}.
 \end{align}

Theorem \ref{assist} can be shown from the following proposition, which 
gives a concrete estimation of $K$ of Theorem \ref{assist}.
\begin{Prop}\label{Ksuf}
Let $\hH_A=\hH_B=\comp^2$.
When  
a non-maximally pure entangled state $\omega$ of $\hH_A\otimes\hH_B$ 
has squared Schmidt coefficients $(r,1-r)$ with $0<r<1/2$
and $k\geq K_s(r):=\max\{K_1(r),K_2(r)\}$, 
for arbitrary two pure entangled states $\psi$ and $\phi$ of $\hH_A\otimes\hH_B$,
there exists a pure state $\omega'$ of $\hH_A\otimes\hH_B$ 
such that (\ref{supplement0}) holds.

 \end{Prop}

Once we prove Proposition \ref{Ksuf}, 
the proof of Theorem \ref{assist} is completed
because it explicitly gives a ``sufficiently large number $K$'' of Theorem \ref{assist}. 
 To prove Proposition \ref{Ksuf}, we should investigate some properties of functions we use as follows.
\begin{Lem}\label{DG}
 \begin{align}
  \frac{v(y)}{h(y)}+2>g(y)
 \end{align}
 holds for $0<y< 1/2$.
\end{Lem}
 \begin{proof}
 In the following, let $y$ satisfy $0<y<1/2$. 
 From $v(y)+2h(y)-h(y)g(y)=h(y)^2+2h(y)-\log y\log (1-y)$ and $h(y)>0$, it is enough to verify $L(y):=h(y)-\log y\log (1-y)\geq 0$.
 Because $\lim_{y\rightarrow 0}L(y)=0$, it is sufficient to show that $L(y)$ is monotonically increasing.
 Since $(1-y)^2y^n-y^2(1-y)^n\leq 0$ holds for $n\geq 2$, we have the following by using the Taylor expansion.
 \begin{align}
  &-(1-y)^2\log (1-y) + y^2\log y\nonumber\\
  =&\sum_{n=1}^{\infty}\frac{1}{n}[(1-y)^2y^n-y^2(1-y)^n]\nonumber\\
  =&(1-y)^2y-y^2(1-y)+\sum_{n=2}^{\infty}\frac{1}{n}[(1-y)^2y^n-y^2(1-y)^n]\nonumber\\
  \geq &(1-y)^2y-y^2(1-y)+\sum_{n=2}^{\infty}[(1-y)^2y^n-y^2(1-y)^n]\nonumber\\
  =&(1-y)^2y-y^2(1-y)+(1-y)y^2-y(1-y)^2\nonumber\\
  =&0.
 \end{align}
 Thus, we obtain that
 \begin{align}
  L'(y)=\frac{-(1-y)^2\log (1-y) + y^2\log y}{y(1-y)}\geq 0,
 \end{align}
  and the proof is completed.
 \end{proof}
 \begin{Lem}\label{Dmono}
  $D(y):=v(y)/h(y)$ is a strictly monotonically decreasing function on $0<y<1/2$.
 \end{Lem}
 \begin{proof}
  From Lemma \ref{DG}, we have
  \begin{align}
   &D'(y)\nonumber\\
  =&\frac{-(\log(1-y)-\log y)[v(y)+2h(y)-h(y)g(y)]}{h(y)^2}<0
  \end{align}
  for $0<y<1/2$, and the proof is completed.
 \end{proof}
 Now, let us prove Proposition \ref{Ksuf}.
 \begin{proof}[Proof of Proposition \ref{Ksuf}]
Thanks to Theorem \ref{noloss}, it is enough to show that for any $0<r<1/2$, the following holds when $k\geq K_s(r)$.
For any $0\leq p,q\leq 1/2$, the following equation about $x$ has a solution with $0\leq x\leq 1/2$, which corresponds to the smaller one of the squared Schmidt coefficients of the byproduct state $\omega'$ in the statement of the lemma.
\begin{align}
 \frac{v(q)+kv(x)}{h(q)+kh(x)}=\frac{v(p)+kv(r)}{h(p)+kh(r)},
\end{align}
or equivalently
\begin{align}
 &f_r(x,p,q,k)\nonumber\\
 :=& (v(p)+kv(r))(h(q)+kh(x))\nonumber\\
 &\hspace{30pt}-(h(p)+kh(r))(v(q)+kv(x))=0.\label{eqcpq}
\end{align}
To show the existence of the solution, by the intermediate-value theorem, it is enough to verify that there exist $x_1$ and $x_2$ such that $f_r(x_1,p,q,k)\leq 0$ and $f_r(x_2,p,q,k)\geq 0$ hold for any $0\leq p,q\leq 1/2$.
Then, we verify this with $x_1=x_r$ and $x_2=1/2$.
After all, it is enough to show that $k\geq K_1(r)$ is sufficient for that $N(p,q):=f_r(x_r,p,q,k)\leq 0$ holds for any $0\leq p,q\leq 1/2$, simultaneously that $k\geq K_2(r)$ is sufficient for that $M(p,q):=f_r(1/2,p,q,k)\leq 0$ holds for any $0\leq p,q\leq 1/2$.

At first, we show that $k\geq K_1(r)$ is sufficient to be $N(p,q)\leq 0$ for any $0\leq p,q\leq 1/2$, which is equivalent to $\max_{0\leq p,q\leq 1/2}N(p,q)\leq 0$.
The first derivative test shows that
$\arg\max_{0\leq q\leq 1/2}N(p,q)=1/2$ for any $0\leq p\leq 1/2$.
Thus, again by the first derivative test, $\max_{0\leq p,q\leq 1/2}N(p,q)=\max_{0\leq p\leq 1/2}N(p,1/2)=N(P_k,1/2)$, where $P_k:=g^{-1}(kv(x_r)/(kh(x_r)+\log 2)+2)$.
Now, we have the following estimation.
\begin{align}
 &N\left(P_k,\frac{1}{2}\right)\nonumber\\
 =&(v(r)h(x_r)-h(r)v(x_r))k^2-kv(x_r)h(P_k)\nonumber\\
 &+(v(P_k)h(x_r)+v(r)\log 2)k+v(P_k)\log 2\nonumber\\
 \leq &(v(r)h(x_r)-h(r)v(x_r))k^2\nonumber\\
 &+(Ch(x_r)+v(r)\log 2)k+C\log 2.\label{X1}
\end{align}
 From Lemma \ref{DG}, we have $g(x_r)=D(r)+2>g(r)$. Thus, since $g$ is strictly monotonically decreasing, $x_r<r$ holds. Then, from Lemma \ref{Dmono}, we obtain $D(x_r)>D(r)$, hence $(v(r)h(x_r)-h(r)v(x_r))<0$.
 Thus, from (\ref{X1}),
 $
  k\geq K_1(r)
 $
  is sufficient for $N(P_k,1/2)\leq 0$.
  
Next, we show that $k\geq K_2(r)$ is sufficient to be $M(p,q):=f_r(1/2,p,q,k)\geq 0$ for any $0\leq p,q\leq 1/2$, which is equivalent to $\min_{0\leq p,q\leq 1/2}M(p,q)\geq 0$.
The first derivative test shows that
$\arg\min_{0\leq p\leq 1/2}M(p,q)=1/2$ for any $0\leq q\leq 1/2$.
Thus, again by the first derivative test, $\min_{0\leq p,q\leq 1/2}M(p,q)=\min_{0\leq q\leq 1/2}M(1/2,q)=M(1/2,Q_k)$, where $Q_k:=g^{-1}(kv(r)/(kh(r)+\log 2)+2)$.
From the estimation
\begin{align}
 &M\left(\frac{1}{2},Q_k\right)\nonumber\\
 =&(\log 2)v(r)k^2-(h(r)k+\log 2) v(Q_k) \nonumber\\
 &\hspace{30pt}+ kv(Q_k)\log 2 +kv(r)h(Q_k)\nonumber\\
 \geq&(\log 2)v(r)k^2-Ch(r)k+C\log 2,
\end{align}
  we find that
 $
 k\geq K_2(r)\label{X}
 $
  is sufficient
 for $M(1/2,Q_k)\geq 0$.

  Therefore, if $k\geq K_s(r)=\max\{K_1(r),K_2(r)\}$, $\max_{0\leq p,q\leq 1/2}N(p,q)\leq 0$ and $\min_{0\leq p,q\leq 1/2}M(p,q)\geq 0$ are simultaneously satisfied, i.e., (\ref{eqcpq}) has a solution for any $0\leq p,q\leq 1/2$, hence the proof is completed.
 \end{proof}

 Furthermore, we give how large $k$ is at least needed for $\omega^{\otimes k}$ to be a universal resource for the asymptotic reversibility as summarized below in Proposition \ref{Kneclem}.
 \begin{Prop}\label{Kneclem}
When a non-maximally pure entangled state $\omega$ of $\hH_A\otimes\hH_B$ 
has squared Schmidt coefficients $(r,1-r)$ with $0<r<1/2$
similar to Proposition \ref{Ksuf} and $k$ is less than $K_0(r)$ defined in \eqref{Knec},
there exists a pair of pure entangled states $\psi$ and $\phi$ 
of $\hH_A\otimes\hH_B$ such that (\ref{supplement0}) does not hold with any pure state $\omega'$ of $\hH_A\otimes\hH_B$. 
 \end{Prop}
This proposition means that 
when the number $k$ is smaller than $K_0(r)$,
there exists a pair of pure entangled states $\psi$ and $\phi$ 
such that the perfect recovery of the conversion $\psi \to \phi$ 
is impossible even asymptotically 
even though we attach the state $\omega^{\otimes k}$.
That is, in order to realize universal supplemental resource from $\omega$,
we need to use $K_0(r)$ tensor product state of $\omega$ at least. 

 \begin{proof}
  Similarly to the proof of Proposition \ref{Ksuf}, we focus on the existence of a solution of (\ref{eqcpq}).
  To prove this proposition, it is sufficient to find a pair of $0\leq p,q\leq 1/2$ such that (\ref{eqcpq}) has no solution when $k<K_0(r)$, which is verified if $\max_{0\leq x\leq 1/2}f_r(x,p,q,k)<0$.
  The first derivative test shows that $\max_{0\leq x\leq 1/2}f_r(x,p,q,k)=f_r(1/2,p,q,k)=M(p,q)$ holds for any $0\leq p,q\leq 1/2$.
  Thus, we show that $M(1/2,x_r)<0$ when $k<K_0(r)$.
  Indeed, since
  \begin{align}
   &M\left(\frac{1}{2},x_r\right)\nonumber\\
   =&\log 2 v(r)k^2\nonumber\\
   &+[(\log 2 -h(r))v(x_r)+v(r)h(x_r)]k-v(x_r)\log 2,
  \end{align}
  $M(1/2,x_r)<0$ holds when $k<K_0(r)$, and the proof is completed.
 \end{proof}

\bibliographystyle{apsrev4-1}
\bibliography{qi}

\end{document}